\documentclass[aps,prb,preprint,preprintnumbers,tightenlines,amsmath,amssymb,noshowkeys,nofootinbib,floatfix]{revtex4-2}
\makeatletter
\def\frontmatter@abstract@produce{%
  \par
  \addvspace{\frontmatter@preabstractspace}%
   \begingroup
    \dimen@\baselineskip
    \setbox\z@\vtop{\unvcopy\absbox}%
    \advance\dimen@-\ht\z@\advance\dimen@-\prevdepth
    \@ifdim{\dimen@>\z@}{\vskip\dimen@}{}%
   \endgroup
   \begingroup
    \prep@absbox
    \unvbox\absbox
    \post@absbox
   \endgroup
  \@ifx{\@empty\mini@notes}{}{\mini@notes\par}%
  \addvspace\frontmatter@postabstractspace
}%
\makeatother
\usepackage{mathtools}
\usepackage{microtype}
\usepackage{amsthm}
\newtheorem{defn}{defn}
\newtheorem{theorem}{Theorem}
\newtheorem{corollary}{Corollary}

\usepackage{derivative}          
\usepackage{braket}              

\newcommand{\rd}{\mathrm{d}}

\newcommand{\acomm}[2]{\{\,#1,#2\,\}}
\newcommand{\slashed}[1]{\ooalign{\hfil/\hfil\cr$#1$\cr}}
\newcommand{\tg}{\tilde{g}}
\newcommand{\ts}{\tilde{s}}

\usepackage{xcolor} 
\usepackage{hyperref}
\hypersetup{
  colorlinks = true,
  linkcolor  = blue!60!black,
  citecolor  = green!50!black,
  urlcolor   = blue!60!black,
  pdftitle   = {The 1+1D N=2 Wess-Zumino model},
  pdfauthor  = {John Dawson}
}

\begin{document}
\title{The 1+1D $\mathcal{N}{=}2$ Wess--Zumino model}
\author{John F. Dawson}
\email[Email: ]{john.dawson@mac.com}
\affiliation{
   Department of Physics,
   University of New Hampshire,
   Durham, NH 03824, USA}
\date{\today}
\begin{abstract}
We develop the 1+1D $\mathcal{N}{=}2$ Wess--Zumino model with a central extension, starting from Wigner's theorem and the supertranslation group.  The unitary representations of the group leads to the SUSY algebra of the generators and an extension of the group to include a central charge $Z$.  We then derive a general boson/fermion mass-degeneracy theorem and construct the covariant derivative and a candidate action realizing this algebra, first for the vacuum ($Z=0$) sector and then for $Z\neq0$, where the central extension's own coordinate $\xi$ is treated explicitly. Compactifying $\xi$ gives a Kaluza--Klein tower of massive multiplets with quantized central charge, while the kink survives intact as the tower's zero mode, saturating the general BPS bound $M\geq|Z|$ derived along the way. This is a pedagogical re-derivation, aimed at a reader's first encounter with these ideas, not a new result.
\end{abstract}
\maketitle

%
%
\section{Introduction}
\label{s:intro}

In ordinary spacetime, invariance under translations means that it does not matter whether an experiment is performed in Seattle or in Moscow, or at what time of day: physics depends only on separations, not on absolute position. This invariance is why momentum and energy are conserved. Superspace extends spacetime by appending a Grassmann (anticommuting) coordinate for each real supercharge --- two of them, $\theta^1,\theta^2$, for the $\mathcal{N}{=}2$ algebra treated here. These are not new directions one could travel to; they are bookkeeping coordinates that carry fermionic content alongside the bosonic $x^\mu$. What ties $\theta^\alpha$ to $x^\mu$ is not that one drags the other along pointwise, but that the \emph{group multiplication law} mixes them: composing two supertranslations, first by $\chi_1^\alpha$ and then by $\chi_2^\alpha$, does not just add the two $\theta$-shifts --- it also forces an ordinary spacetime translation, quadratic in $\chi_1,\chi_2$ (\S\ref{ss:composition}). That is the precise, coordinate-level sense in which supersymmetry is a ``square root'' of spacetime translations.

The supertranslation group acts on $(x^\mu,\theta^\alpha)$ together, so a supertranslation is a translation along $\theta$ as well. Concretely, ``moving along $\theta$'' means converting bosonic content into fermionic content and back, exactly as $\mathcal{Q}_\alpha$ is found to do in \S\ref{ss:supermultiplets}. The direct analogue of energy-momentum conservation is that $\mathcal{Q}_\alpha$ itself, the Noether charge of this $\theta$-translation invariance, is conserved. Because $\mathcal{Q}_\alpha$ also commutes with $\mathcal{P}_\mu$, converting a state's statistics without touching its momentum, this conservation carries an immediate corollary: a boson and its fermion partner --- related by $\mathcal{Q}_\alpha$ --- are forced to share the same mass (\S\ref{ss:supermultiplets}). It does not matter whether one measures the boson or its superpartner --- the mass that comes back is the same.

The central charge $Z$ that appears once this algebra is extended plays two distinct roles in what follows, and the paper is organized to keep them visibly separate. The first is purely algebraic: before any Lagrangian is written down, \S\ref{s:superspace} derives the extended algebra from group theory alone and shows $Z$ is itself an independent Casimir of that algebra, on equal footing with $P^\mu P_\mu$, entering the mass-degeneracy theorem that constrains the spectrum of \emph{any} unitary representation. The second role only emerges once a specific action is chosen: \S\ref{s:coDer} builds the covariant derivative carrying the algebra into dynamics, and \S\ref{s:action} then realizes $Z$ concretely in two cases side by side --- a vacuum sector with $Z=0$ (\S\ref{ss:vacuum}), where the component Lagrangian, equations of motion, and mass spectrum all fall out of a single superspace equation; and a $Z\neq0$ sector (\S\ref{ss:Znotzero}) where the central extension's own superspace coordinate $\xi$ is taken seriously rather than discarded, so that $Z$ becomes the topological charge of a kink saturating the BPS bound $M\geq|Z|$ and, on compactifying $\xi$, the quantized Kaluza--Klein momentum organizing a tower of massive multiplets. Three appendices collect supporting material used throughout rather than interrupting the main line of argument: Appendix~\ref{s:notation} fixes notation and conventions, Appendix~\ref{s:grassmann} reviews the Grassmann calculus these derivations lean on, and Appendix~\ref{s:IBPextended} proves the integration-by-parts lemma \S\ref{ss:Znotzero} needs for the extended superspace coordinate $\xi$.

None of this is new. The general construction --- a chiral superfield with a superpotential --- is Wess and Zumino's, written down in four dimensions more than fifty years ago~\cite{WessZumino1974}; its central charge and the BPS kink that saturates it are Witten and Olive's, from almost as long ago~\cite{WittenOlive1978}; Gates, Grisaru, Ro\v{c}ek, and Siegel's superspace treatise is already over four decades old~\cite{GatesGrisaruRocekSiegel1983}; and Wess and Bagger's textbook account has itself now passed thirty years~\cite{WessBagger1992}. What follows is a pedagogical re-derivation, not a new result --- but building the algebra up from Wigner's theorem and the supertranslation group alone, with no Lagrangian in sight until \S\ref{s:action}, may be a clearer route into what supersymmetry actually \emph{is} than the more common ``write down an invariant Lagrangian and verify it'' approach.

%
%
\section{1+1D $\mathcal{N}=2$ Superspace}
\label{s:superspace}

We develop the $\mathcal{N}{=}2$ supersymmetry algebra for the general case of real supercharges in 1+1 dimensions.  Superspace consists of
\begin{equation}\label{e:superspace-def}
   s = ( x^\mu,\theta^{\alpha} )\>,\quad
   x^{\mu} = (t,x)\>,\quad
   \theta^{\alpha} = (\theta^1,\theta^2) \>,
\end{equation}
where $(t,x)$ are real variables and $(\theta^1,\theta^2)$ are real Grassmann variables.  We choose a metric with the signature $(-,+)$.  Our notation is explained in Appendix~\ref{s:notation} and properties of Grassmann variables are discussed in Appendix~\ref{s:grassmann}.

%
%
\subsection{The supertranslation group and its composition rule}
\label{ss:composition}

A SUSY translation is labeled by a pair $g=(a^\mu,\chi^\alpha)$: a spacetime translation $a^\mu$ together with a constant Grassmann spinor $\chi^\alpha$. It acts on superspace coordinates $s = (x^\mu,\theta^\alpha)$ by
\begin{equation}\label{e:groupaction}
   g: \qquad 
   x^\mu \to x^\mu + a^\mu - i\,\bar\chi\gamma^\mu\theta\,,
   \qquad \theta^\alpha\to\theta^\alpha+\chi^\alpha\,.
\end{equation}
where $\theta$ and $\chi$  are two-component Majorana spinors with $\bar{\chi} = \chi^T C$, and where $C$ is the charge-conjugate operator.  Our conventions for the $\gamma^{\mu}$ and associated matrices are given in detail in Appendix~\ref{s:notation}.  In component form, \eqref{e:groupaction} becomes:
\begin{equation}\label{e:groupaction-components}
\begin{aligned}
   t &\to t + a^0 + i\,(\chi^1\theta^1+\chi^2\theta^2)\,, \\
   x &\to x + a^1 - i\,(\chi^1\theta^1-\chi^2\theta^2)\,,
\end{aligned}
\end{equation}
The explicit $i$ multiplying each bracket is required for $t,x$ to
stay real: $\chi^{\alpha}\theta^{\alpha}$ are bilinears of two
\emph{independent} real Grassmann quantities and are anti-Hermitian, so that $i \chi^{\alpha}\theta^{\alpha}$ is real.  

This supertranslation transformation forms a group.  To show this, apply $g_1=(a_1,\chi_1)$ first and then $g_2=(a_2,\chi_2)$ to a generic point in superspace:
\begin{equation}\label{e:g1g2g3}
\begin{aligned}
   x^\mu
   &\xrightarrow{\ g_1\ } 
      x^\mu + a_1^\mu - i\bar\chi_1\gamma^\mu\theta 
   \xrightarrow{\ g_2\ }
      x^\mu + a_1^\mu - i\bar\chi_1\gamma^\mu\theta
            + a_2^\mu - i\bar\chi_2\gamma^\mu(\theta+\chi_1) \\
   &\quad\quad = x^\mu+(a_1+a_2)^\mu - i\bar\chi_2\gamma^\mu\chi_1
             - i(\bar\chi_1 + \bar\chi_2)\gamma^\mu\theta
   \equiv x^{\mu}_3 - i \bar\chi_3 \gamma^{\mu} \theta \>.
\end{aligned}
\end{equation}
So the \textbf{composition rule} of the supertranslation group is given by:
\begin{equation}\label{e:compositionrule}
    a_3^\mu = a_1^\mu + a_2^\mu - i\,\bar\chi_2\gamma^\mu\chi_1\,,
   \qquad
   \chi_3^\alpha = \chi_1^\alpha + \chi_2^\alpha\,.
\end{equation}
Associativity follows by applying \eqref{e:groupaction} three times; we omit the bookkeeping.

%
%
\subsection{Unitary representation on the Hilbert space of fields}
\label{ss:unitaryrep}

In superspace coordinates, a state vector $\ket{\Psi}$ is given by the inner product: $\Psi(s) = \braket{s|\Psi}$.  A SUSY transformation $g$ that takes $s$ to $s'$ is written as $s'=g(s)$.  Then $\Psi(s') = \braket{s'|\Psi} = \braket{g(s)|\Psi} = \braket{s|\Psi(g)}$.  Now let U(g) be the SUSY which takes a vector $\ket{\Psi}$ in the ray $\mathcal{R}$ to a vector $\ket{\Psi(g)}$ in the ray $\mathcal{R}(g)$ as:
\begin{equation}\label{e:UPsitoPsi}
   \ket{\Psi(g)} = U(g) \, \ket{\Psi} \>.
\end{equation}
The state vector in superspace coordinates transforms then as:
\begin{equation}\label{e:Psi-s-rep}
   \Psi(s') = \braket{s| U(g) | \Psi} \>.
\end{equation}
This representation property does not depend on the specific coordinate convention used above to motivate $\ket{\Psi(g)}$: physically, performing the symmetry transformation $g_1$ and then $g_2$ must be the same operation on the system as performing the single transformation $g_2\circ g_1$ directly, so the operators implementing them can differ by no more than a phase. By Wigner's theorem, that phase is all the freedom there is --- $U(g)$ need only furnish a \emph{projective} representation of the composition rule \eqref{e:compositionrule}, $U(g_2)U(g_1)=e^{i\phi(g_2,g_1)}U(g_2\circ g_1)$, with the phase $\phi$ constrained by a cocycle condition:
\begin{equation}\label{e:cocycle}
   \phi(g_2,g_1)+\phi(g_3,g_2\circ g_1) = \phi(g_3,g_2)+\phi(g_3\circ g_2,g_1)\,.
\end{equation}
The composition rule \eqref{e:compositionrule} shows that $a^\mu$ enters only through ordinary addition, together with one term quadratic in $\chi$; any dependence of $\phi$ on $a$ alone would just be a phase on translations, removable by redefining $U(g)\to e^{iK_\mu a^\mu}U(g)$, so we look for a $\phi$ built purely from $\chi_1,\chi_2$. Each component of $\chi^\alpha$ is Grassmann and squares to zero, so the lowest nontrivial choice is bilinear, giving the ansatz:
\begin{equation}\label{e:Mphase}
   \phi(g_2,g_1)=i\,\chi_2^\alpha M_{\alpha\beta}\chi_1^\beta \>,
\end{equation}
as can easily be checked.  Also in order for the phase to be real, $M$ must be real.  Next write $M$ as the sum of diagonal, symmetric-traceless, and antisymmetric parts, 
\begin{equation}\label{e:M-symmetric-traceless}
\begin{aligned}
   M = \begin{pmatrix} K_0 + K_1 & Z+\kappa \\ Z-\kappa & K_0 - K_1 \end{pmatrix}
   &=
   K_0 \sigma_0 + K_1 \sigma_3 + i \kappa\, \sigma_2 + Z\, \sigma_1 \\
   &=
   - K_0 C \gamma^0 + K_1 C \gamma^1 + \kappa \, \gamma^0 + Z\, \gamma^1 \>,
\end{aligned}
\end{equation}
so that (\ref{e:Mphase}) becomes
\begin{equation}\label{e:Phase-value}
   \phi(g_2,g_1)
   = i\, K_1 \, (\bar{\chi}_2\, \gamma^{1}\, \chi_1)
     - i\, K_0 \, (\bar{\chi}_2\, \gamma^{0}\, \chi_1)
     + i\, \kappa \, (\bar\chi_2\, \chi_1) + i\, Z \, (\chi_2^T\,\gamma^1\, \chi_1) \>.
\end{equation}
But by the composition rule \eqref{e:compositionrule}, the $\kappa$ term can be written as
\begin{equation}\label{e:Phase-remove1}
   i\, \kappa \, (\bar\chi_2\, \chi_1)
   =
   \frac{i\,\kappa}{2} \, [\, (\bar\chi_3\, \chi_3) - (\bar\chi_1\, \chi_1) - (\bar\chi_2\, \chi_2) \,] \>,
\end{equation}
which is a trivial phase that can be removed by a phase change.  Similarly, the $K_0,K_1$ terms in \eqref{e:Phase-value} become:
\begin{equation}\label{e:phase1-remove2}
   K_1\, (\, a_1^1 + a_2^1 - a_3^{1} \,) - K_0\, (\, a_1^0 + a_2^0 - a_3^{0} \,) \>.
\end{equation}
which can be compensated for by shifting the momentum components, $P_0\to P_0+K_0$ and $P_1\to P_1-K_1$, via the phase redefinition:
\begin{equation}\label{e:Pa-phase}
   U(g) \to e^{i\,(K_0\, a^0 - K_1\, a^1)} \, U(g) \>.
\end{equation}
So from \eqref{e:Phase-value} we are left with
\begin{equation}\label{e:phase-final}
   \phi(g_2,g_1) = i\, Z \, (\chi_2^T\,\gamma^1\, \chi_1) \>.
\end{equation}

%
%
\subsection{Extension of the SUSY group}
\label{ss:SUSY-extension}

The phase $\phi(g_2,g_1)$ in \eqref{e:phase-final} is linear in the central charge $Z$ which suggests promoting $Z$ to a group generator.  We introduce a new real group parameter $\zeta$ which is conjugate to the central charge $Z$, and extend the SUSY group with the additional element: $\tg = (a^{\mu},\chi^{\alpha},\zeta)$.  In addition to the transformation rules in \eqref{e:compositionrule}, the extended group composition rules are now given by:
\begin{equation}\label{e:tg-composition}
\begin{aligned}
   a_3^\mu &= a_1^\mu + a_2^\mu - i\,\bar\chi_2\gamma^\mu\chi_1 \,,
   \\
   \zeta_3 &= \zeta_2 + \zeta_1 + i \, (\chi_2^T \gamma^1 \chi_1) \,,
   \\
   \chi_3^\alpha &= \chi_1^\alpha + \chi_2^\alpha\,.
\end{aligned}
\end{equation}
Superspace is also extended to include an addition real coordinate $\xi$ which now becomes: $\ts = (x^{\mu},\theta^{\alpha},\xi)$.  Now when a SUSY translation acts on $\tilde{s}$, the result is:
\begin{equation}\label{e:SUSY-extended-translation}
\begin{aligned}
   x^\mu &\to x^{\mu} + a^{\mu} - i \bar{\chi} \gamma^{\mu} \theta \,,
   \\
   \xi &\to \xi + \zeta + i \, \chi^T \gamma^1 \theta \,,
   \\
   \theta^\alpha &\to \theta^\alpha + \chi^\alpha\,.
\end{aligned}
\end{equation}
Since $\theta,\chi$ each carry mass dimension $-\tfrac12$ (so that $\bar\chi\gamma^\mu\theta$, appearing in $\delta x^\mu$ above, has the same dimension as $x^\mu$ itself), the matching term $\chi^T\gamma^1\theta$ in $\delta\xi$ forces $\xi$ to carry dimension $-1$ as well --- a length, exactly like $x^\mu$, not a dimensionless angle. This is worth flagging now: it means compactifying $\xi$ later (\S\ref{ss:Znotzero}) needs an honest length scale, since $\xi$ is dimensionful in precisely the same way $x$ is, not a free-floating group parameter that happens to look angle-like.

A simple calculation shows that the composition rules \eqref{e:tg-composition} are satisfied by this extended superspace.  Infinitesimal changes of the superspace coordinates are generated by:
\begin{equation}\label{e:SUSY-extended-infintesimal}
\begin{aligned}
   \Delta x^\mu &= \Delta a^{\mu} - i \Delta \bar{\chi} \gamma^{\mu} \theta \,,
   \\
   \Delta \xi &= \Delta \zeta + i \, \Delta \chi^T \gamma^1 \theta \,,
   \\
   \Delta \theta^\alpha &= \Delta \chi^\alpha\,.
\end{aligned}
\end{equation} 

The most general function in this superspace is of the form,
\begin{equation}\label{e:Phi-def}
\begin{aligned}
   \Phi(\ts)
   &= \phi(x,\xi) + \bar\theta\psi(x,\xi) + \tfrac{i}{2}\,\bar\theta\theta\,F(x,\xi) \,,
   \\
   &= \phi(x,\xi) + [\, \theta^1\psi^2(x,\xi)-\theta^2\psi^1(x,\xi) \,]
      + i\,\theta^1\theta^2\,F(x,\xi) \>,
\end{aligned}
\end{equation}
where $\phi(x,\xi)$ and $F(x,\xi)$ are real functions, and $\psi(x,\xi)$ is a real Majorana Grassmann spinor.

%
%
\subsection{Supersymmetric generators}
\label{ss:generators}

For infinitesimal $\tg$, define the generators of the supersymmetric algebra $P_\mu$ (momentum) and $Q_\alpha$ (supercharge) and $Z$ (central charge) by the relation:
\begin{equation}\label{e:Udef}
   U(1+\Delta \tilde{g}) 
   = 
   1 + i\,\Delta a^\mu P_\mu + \Delta\chi^\alpha Q_\alpha + i\,\Delta \zeta\, Z \cdots\,,
\end{equation}
where $P_{\mu}$ and $Z$ are Hermitian (real, commuting) operators, while $Q_\alpha$ and $\chi^\alpha$ are both \emph{real Grassmann} quantities (Hermitian and anticommuting).  The group has now been extended to incorporate the phase factors.  The group transformation rule now reads:
\begin{equation}\label{e:U1U2toU3}
   U(\tg_1) \, U(\tg_2) = U(\tg_1 \tg_2) \>,
\end{equation}
with no phase factor needed.

To find the algebra of the generators, we consider the following composite transformation:
\begin{equation}\label{e:UinvDUUtoDU}
   U(\tg^{-1}_1) U(1 + \Delta \tg_2) U(\tg_1) = U(1 + \Delta \tg_3) \>,
   \quad
   \Delta \tg_3 = \tg^{-1}_1 \Delta \tg_2 \, \tg_1 \>.
\end{equation}
where 
\begin{equation}\label{e:ggpdg}
\begin{aligned}
   \tg_1 &= (a^{\mu}_1,\chi^{\alpha}_1,\zeta_1) \>,
   \\
   \tg^{-1}_1 &= (-a^{\mu}_1, -\chi^{\alpha}_1,-\zeta_1) \>,
   \\
   \Delta \tg_2 &= (\Delta a^{\mu}_2,\Delta\chi^{\alpha}_2,\Delta\zeta_2) \>.
\end{aligned}   
\end{equation}
Then using the extended group composition rules \eqref{e:tg-composition}, after a little work, we find:
\begin{equation}\label{e:ggg}
\begin{aligned}
   \Delta a^{\mu}_3 &= \Delta a_2^\mu + 2i\,\bar\chi_1\gamma^\mu\Delta\chi_2
   \\
   \Delta \chi^{\alpha}_3 &= \Delta\chi_2^\alpha
   \\
   \Delta \zeta_3 &= \Delta\zeta_2 + 2i\,(\Delta\chi_2^T\gamma^1\chi_1)
\end{aligned}
\end{equation}

So then \eqref{e:UinvDUUtoDU} becomes:
\begin{equation}\label{e:ggU}
\begin{aligned}
   &1
   +
   i\,\Delta a^\mu_2 U(\tg_1^{-1}) P_\mu U(\tg_1)
   +
   \Delta\chi^\alpha_2 U(\tg_1^{-1}) Q_\alpha U(\tg_1)
   +
   i\,\Delta \zeta_2\, U(\tg_1^{-1}) Z U(\tg_1) \cdots\,,
   \\
   & \qquad
   =
   1
   +
   i\,\Delta a^\mu_3 P_\mu
   +
   \Delta\chi^\alpha_3 Q_\alpha
   +
   i\,\Delta \zeta_3 \, Z \cdots\,,
\end{aligned}
\end{equation}

Substituting \eqref{e:ggg} into the right-hand side of \eqref{e:ggU} and comparing the coefficients of the independent infinitesimal parameters $\Delta a_2^\mu,\Delta\chi_2^\alpha,\Delta\zeta_2$ on the two sides gives the transformation law of each generator under conjugation by the finite element $\tg_1$.

\begin{equation}\label{e:PZQtransform}
\begin{aligned}
   U(\tg_1^{-1})\,P_\mu\,U(\tg_1) &= P_\mu\,,
   \\
   U(\tg_1^{-1})\,Z\,U(\tg_1) &= Z\,,
   \\
   U(\tg_1^{-1})\,Q_\alpha\,U(\tg_1) &= Q_\alpha + 2\,(C\gamma^\mu\chi_1)_\alpha\,P_\mu - 2\,(\gamma^1\chi_1)_\alpha\,Z\,.
\end{aligned}
\end{equation}
Using $C\gamma^0=-I$, $C\gamma^1=\sigma_3$, and $\gamma^1=\sigma_1$ (Appendix~\ref{s:notation}), the $Q_\alpha$ result reads component by component as
\begin{equation}\label{e:Qtransform-components}
\begin{aligned}
   U(\tg_1^{-1})\,Q_1\,U(\tg_1) &= Q_1 - 2\chi_1^1\,(P_0-P_1) - 2\chi_1^2\,Z\,,
   \\
   U(\tg_1^{-1})\,Q_2\,U(\tg_1) &= Q_2 - 2\chi_1^2\,(P_0+P_1) - 2\chi_1^1\,Z\,.
\end{aligned}
\end{equation}

The extended SUSY algebra now follows by taking $\tg_1$ itself infinitesimal, $\tg_1\to 1+\Delta\tg_1$, in \eqref{e:PZQtransform}--\eqref{e:Qtransform-components}, and expanding the left side with \eqref{e:Udef} again: $U(1+\Delta\tg_1)^{-1}XU(1+\Delta\tg_1)\approx X+[X,\,i\Delta a_1^\nu P_\nu+\Delta\chi_1^\beta Q_\beta+i\Delta\zeta_1 Z\,]$. Matching coefficients of $\Delta a_1^\nu,\Delta\chi_1^\beta,\Delta\zeta_1$ against the (exact, $\chi_1$-linear) right-hand sides of \eqref{e:PZQtransform} gives
\begin{equation}\label{e:SUSYalgebra-trivial}
   [P_\mu,P_\nu] = 0\,,
   \qquad
   [P_\mu,Q_\alpha] = 0\,,
   \qquad
   [P_\mu,Z] = 0\,,
   \qquad
   [Q_\alpha,Z] = 0\,,
\end{equation}
from the (trivial) $P_\mu,Z$ lines, while the $Q_\alpha$ line --- using $[Q_\alpha,\Delta\chi_1^\beta Q_\beta]=-\Delta\chi_1^\beta\{Q_\alpha,Q_\beta\}$, since $\Delta\chi_1^\beta$ is a Grassmann number anticommuting with the Grassmann-odd $Q_\alpha$ --- gives the one nontrivial relation:
\begin{equation}\label{e:SUSYalgebra}
   \boxed{\;\{Q_\alpha,Q_\beta\} = -2\,(C\gamma^\mu)_{\alpha\beta}\,P_\mu + 2\,(\gamma^1)_{\alpha\beta}\,Z\;}\,.
\end{equation}
(The $[P_\mu,Q_\alpha]=0$ and $[Q_\alpha,Z]=0$ relations reappear here too, as the vanishing $\Delta a_1^\nu$ and $\Delta\mu_1$ coefficients on the right of \eqref{e:PZQtransform}, consistently.) Using $C\gamma^0=-I$, $C\gamma^1=\sigma_3$, $\gamma^1=\sigma_1$, this reads in components
\begin{equation}\label{e:SUSYalgebra-components}
   Q_1^2 = P_0-P_1\,,
   \qquad
   Q_2^2 = P_0+P_1\,,
   \qquad
   \{Q_1,Q_2\} = 2Z\,,
\end{equation}
the standard 2D $\mathcal{N}{=}2$ SUSY algebra with central extension: $Q_1^2+Q_2^2=2P_0$ is twice the Hamiltonian, $Q_2^2-Q_1^2=2P_1$ is twice the spatial momentum, and $\{Q_1,Q_2\}=2Z$ is the central term.

Since $P_\mu$ and $Z$ commute with every generator in \eqref{e:SUSYalgebra-trivial}--\eqref{e:SUSYalgebra}, this algebra has two Casimirs,
\begin{equation}\label{e:Casimirs}
   P^\mu P_\mu
   \qquad\text{and}\qquad
   Z\,,
\end{equation}
and no independent Casimir can be built from the fermionic sector: by \eqref{e:SUSYalgebra-components}, every symmetric bilinear in $Q_\alpha$ already reduces to $P_\mu$ or $Z$, and the one remaining bilinear, the commutator $[Q_1,Q_2]$, fails to commute with $Q_1$ itself.

Since $Q_1^2+Q_2^2=2P_0$ is a sum of squares of Hermitian operators, $P_0\geq0$ on every state, confirming that $P_0$ is the Hamiltonian. With $\eta^{\mu\nu}=\mathrm{diag}(-1,1)$, the first Casimir is $P^\mu P_\mu=-P_0^2+P_1^2$; defining
\begin{equation}\label{e:Mdef}
   M^2 \equiv -P^\mu P_\mu = P_0^2-P_1^2\,,
\end{equation}
a one-particle state of energy $E=P_0$ and momentum $p=P_1$ satisfies $E^2-p^2=M^2$, the standard dispersion relation.

%
%
\subsection{Differential representation of SUSY operators}
\label{ss:DiffRepsSUSY}

For infinitesimal transformations of the extended superspace, from the representation given in \eqref{e:Psi-s-rep}, we have
\begin{equation}\label{e:Psi-DeltaU-transform}
\begin{aligned}
   \Psi(\ts + \Delta\ts) 
   &=
   \braket{\ts|U(1 + \Delta \tg + \cdots) |\Psi}
   \\
   &=
   \braket{\ts|
   1 + i\,\Delta a^\mu P_\mu + \Delta\chi^\alpha Q_\alpha + i\,\Delta \zeta\, Z \cdots | \Psi}
   \\
   &=
   \Psi(\ts) 
   +
   i \Delta a^\mu \braket{\ts|P_{\mu}|\Psi}
   +
   \Delta\chi^\alpha \braket{\ts|Q_{\alpha}|\Psi}
   +
   i\,\Delta \zeta\, \braket{\ts|Z|\Psi}
   +
   \cdots \>.
\end{aligned}
\end{equation}
Now using \eqref{e:SUSY-extended-infintesimal}, the infinitesimal change in $\Psi(x,\theta,\xi)$ is given by
\begin{equation}\label{e:derivesPsi}
\begin{split}
   &\pdv{\Psi(x,\theta,\xi)}{x^{\mu}} \Delta x^{\mu} 
   + 
   \pdv{\Psi(x,\theta,\xi)}{\theta^{\alpha}} \Delta \theta^{\alpha}
   +
   \pdv{\Psi(x,\theta,\xi)}{\xi} \Delta \xi 
   \\
   &=
   \pdv{\Psi(x,\theta,\xi)}{x^{\mu}} \,
   [\, \Delta a^{\mu} - i \Delta \bar{\chi} \gamma^{\mu} \theta \,]
   +
   \pdv{\Psi(x,\theta,\xi)}{\theta^{\alpha}} \Delta \chi^\alpha
   +
   \pdv{\Psi(x,\theta,\xi)}{\xi} \,
   [\, \Delta \zeta + i \, \Delta \chi^T \gamma^1 \theta \,]
\end{split}
\end{equation}
Comparing coefficients of $\Delta a^{\mu}$, $\Delta \chi^{\alpha}$ and $\Delta \zeta$ gives:
\begin{equation}\label{e:PQZ-diff-ops}
\begin{aligned}
   \mathcal{P}_{\mu} \Psi(x,\theta,\xi)
   & \equiv
   \braket{x,\theta,\xi|P_{\mu}|\Psi}
   =
   -i\,\partial_\mu \Psi(x,\theta,\xi)\,,
   \\
   \mathcal{Q}_{\alpha} \Psi(x,\theta,\xi)
   & \equiv
   \braket{x,\theta,\xi|Q_{\alpha}|\Psi}
   =
   \Bigl[\,\partial_\alpha - i\,(C\gamma^\mu\theta)_\alpha\,\partial_\mu + i\,(\gamma^1\theta)_\alpha\,\partial_\xi\,\Bigr]\Psi(x,\theta,\xi)\,,
   \\
   \mathcal{Z} \Psi(x,\theta,\xi)
   & \equiv
   \braket{x,\theta,\xi|Z|\Psi}
   =
   -i\,\partial_\xi \Psi(x,\theta,\xi)\,.
\end{aligned}
\end{equation}

%
%
\subsection{Supermultiplets}
\label{ss:supermultiplets}

The most general wavefunction describing the state of the system is given by an expansion of the form,
\begin{equation}\label{e:Psi-general-form}
   \Psi(x,\theta,\xi) 
   = 
   \phi(x,\xi) + \bar\theta\psi(x,\xi) + \tfrac{i}{2}\,\bar\theta\theta\,F(x,\xi)
\end{equation}
Then the action of the differential operators \eqref{e:PQZ-diff-ops} on $\Psi(x,\theta,\xi)$ is given by:

Since $\partial_\mu$ and $\partial_\xi$ differentiate each component in place without touching $\theta$, $\mathcal{P}_\mu$ and $\mathcal{Z}$ act diagonally on \eqref{e:Psi-general-form}:
\begin{equation}\label{e:PZ-on-Psi}
\begin{aligned}
   \mathcal{P}_\mu\Psi &= -i\,\partial_\mu\phi + \bar\theta\,(-i\,\partial_\mu\psi) + \tfrac{i}{2}\,\bar\theta\theta\,(-i\,\partial_\mu F)\,,
   \\
   \mathcal{Z}\Psi &= -i\,\partial_\xi\phi + \bar\theta\,(-i\,\partial_\xi\psi) + \tfrac{i}{2}\,\bar\theta\theta\,(-i\,\partial_\xi F)\,:
\end{aligned}
\end{equation}
the bosonic content ($\phi,F$) and fermionic content ($\psi$) are each simply differentiated in place; $\mathcal{P}_\mu,\mathcal{Z}$ never mix the two.

As a check, applying $\mathcal{P}_\mu$ \eqref{e:PZ-on-Psi} twice and contracting with $\eta^{\mu\nu}=\mathrm{diag}(-1,1)$,
\begin{equation}\label{e:PP-on-Psi}
   \mathcal{P}_\mu\mathcal{P}^\mu\Psi = \eta^{\mu\nu}(-i\partial_\mu)(-i\partial_\nu)\Psi = (\partial_0^2-\partial_1^2)\,\Psi\,,
\end{equation}
true for any $\Psi$, and again diagonal on $\phi,\psi,F$ since $\mathcal{P}_\mu$ is. Restricting to a state of definite mass, $\mathcal{P}_\mu\mathcal{P}^\mu\Psi=-M^2\Psi$ with $M^2\equiv-P^\mu P_\mu$ \eqref{e:Mdef}, \eqref{e:PP-on-Psi} becomes
\begin{equation}\label{e:KG-on-Psi}
   \bigl(\partial_0^2-\partial_1^2+M^2\bigr)\Psi = 0\,,
\end{equation}
the Klein--Gordon equation, holding identically for every component of $\Psi$ at once --- it falls out directly from the Casimir, with no Lagrangian needed, and applies uniformly across the whole supermultiplet rather than component by component.

$\mathcal{Q}_\alpha$ is different. A direct computation from \eqref{e:PQZ-diff-ops} and \eqref{e:Psi-general-form}, discarding every term with a repeated $\theta$, gives
\begin{equation}\label{e:Q-on-Psi-compact}
\begin{aligned}
   \mathcal{Q}_1\Psi &= \psi^2 + i\theta^1(\partial_0-\partial_1)\phi + i\theta^2(F+\partial_\xi\phi) - i\theta^1\theta^2\bigl[(\partial_0-\partial_1)\psi^1+\partial_\xi\psi^2\bigr]\,,
   \\
   \mathcal{Q}_2\Psi &= -\psi^1 + i\theta^1(\partial_\xi\phi-F) + i\theta^2(\partial_0+\partial_1)\phi - i\theta^1\theta^2\bigl[(\partial_0+\partial_1)\psi^2+\partial_\xi\psi^1\bigr]\,.
\end{aligned}
\end{equation}
Reading off the coefficient landing in each slot of $\Psi$'s own expansion \eqref{e:Psi-general-form} gives the dictionary
\begin{equation}\label{e:Q1-on-Psi}
\begin{aligned}
   \mathcal{Q}_1\Psi\big|_{\phi} &= \psi^2\,,
   &
   \mathcal{Q}_1\Psi\big|_{\psi^1} &= -iF - i\,\partial_\xi\phi\,,
   \\
   \mathcal{Q}_1\Psi\big|_{\psi^2} &= i(\partial_0-\partial_1)\phi\,,
   &
   \mathcal{Q}_1\Psi\big|_{F} &= -(\partial_0-\partial_1)\psi^1 - \partial_\xi\psi^2\,,
\end{aligned}
\end{equation}
\begin{equation}\label{e:Q2-on-Psi}
\begin{aligned}
   \mathcal{Q}_2\Psi\big|_{\phi} &= -\psi^1\,,
   &
   \mathcal{Q}_2\Psi\big|_{\psi^1} &= -i(\partial_0+\partial_1)\phi\,,
   \\
   \mathcal{Q}_2\Psi\big|_{\psi^2} &= -iF + i\,\partial_\xi\phi\,,
   &
   \mathcal{Q}_2\Psi\big|_{F} &= -(\partial_0+\partial_1)\psi^2 - \partial_\xi\psi^1\,,
\end{aligned}
\end{equation}
where $\mathcal{Q}_\alpha\Psi|_X$ denotes the coefficient that lands in slot $X$. (Setting $\partial_\xi\to0$ reduces \eqref{e:Q-on-Psi-compact}--\eqref{e:Q2-on-Psi} to the SUSY transformation used in \S\ref{ss:Znotzero}, as it must.) Every entry sends bosonic content ($\phi,F$) to fermionic ($\psi$) or fermionic to bosonic --- never $\phi\to F$ or $\psi^1\to\psi^2$ directly --- while $\mathcal{P}_\mu,\mathcal{Z}$ \eqref{e:PZ-on-Psi} never cross that line at all.

Since $[\mathcal{P}_\mu,\mathcal{Q}_\alpha]=[\mathcal{Q}_\alpha,\mathcal{Z}]=0$ (the same relations as \S\ref{ss:generators}, now realized as differential operators), acting with $\mathcal{Q}_\alpha$ leaves $p,z$ untouched, so each boson/fermion pair identified in \eqref{e:Q1-on-Psi}--\eqref{e:Q2-on-Psi} --- $\phi\leftrightarrow\psi^2$, $\psi^1\leftrightarrow F$, and so on --- are superpartners with the same $p,z$, hence, by \eqref{e:KG-on-Psi}, the same mass $M$. This is the full content of the boson/fermion mass-degeneracy theorem, and it holds for \emph{any} unitary representation of the algebra \eqref{e:SUSYalgebra-trivial}--\eqref{e:SUSYalgebra} --- no action or Lagrangian has been used.

%
%
\section{The covariant derivative}
\label{s:coDer}

We first need to find out how the differentials transform under the extended SUSY.  Differentials of Eqs.~\eqref{e:SUSY-extended-translation} give the following:
\begin{equation}\label{e:dxptox}
\begin{aligned}
   \rd x^{\mu\prime} 
   &=
   \rd x^{\mu} - i\,\bar{\chi} \gamma^{\mu} \rd \theta \,,
   \\
   \rd \xi' 
   &=
   \rd \xi + i \, \chi^{T} \gamma^1 \rd \theta \,,
   \\
   \rd \theta'
   &=
   \rd \theta \,.
\end{aligned}
\end{equation}
To find how the derivatives transform, where $\partial_\mu = \frac{\partial}{\partial x^\mu}$, first invert the coordinate transformation \eqref{e:SUSY-extended-translation} exactly. Substituting $\theta=\theta'-\chi$ into the $x^\mu$ and $\xi$ relations and using $\bar\chi\gamma^\mu\chi=0$ and $\chi^T\gamma^1\chi=0$ (the same symmetric-matrix self-contraction identities used to verify $\tg_1^{-1}$ in \S\ref{ss:generators}), the cross term drops out and the inverse is exact:
\begin{equation}\label{e:inverse-coords}
   x^\mu = x'^\mu - a^\mu + i\,\bar\chi\gamma^\mu\theta'\,,
   \qquad
   \xi = \xi' - \zeta - i\,\chi^T\gamma^1\theta'\,,
   \qquad
   \theta^\alpha = \theta'^\alpha-\chi^\alpha\,.
\end{equation}
Since $a^\mu,\chi^\alpha,\zeta$ are held fixed, $x^\mu$ and $\xi$ above depend on $x'^\mu$ (resp.\ $\xi'$) only through the identity term, so
\begin{equation}\label{e:dxdxi-transform}
   \partial'_\mu = \partial_\mu\,,
   \qquad
   \partial'_\xi = \partial_\xi\,.
\end{equation}
For $\partial'_\alpha\equiv\partial/\partial\theta'^\alpha$, the left-acting derivative must anticommute past the single power of $\chi$ in each coefficient of \eqref{e:inverse-coords} before it reaches $\theta'^\beta$, picking up a minus sign; using the symmetry $(C\gamma^\mu)^T=C\gamma^\mu$ to write $(\bar\chi\gamma^\mu)_\alpha=(C\gamma^\mu\chi)_\alpha$, this gives
\begin{equation}\label{e:dtheta-transform}
   \boxed{\;\partial'_\alpha = \partial_\alpha - i\,(C\gamma^\mu\chi)_\alpha\,\partial_\mu + i\,(\gamma^1\chi)_\alpha\,\partial_\xi\;}\,.
\end{equation}
So $\partial_\mu$ and $\partial_\xi$ are unaffected by the SUSY transformation, and only the Grassmann derivative picks up a correction --- exactly the structure needed to build a genuine supercovariant derivative $D_\alpha$ in this extended superspace.

\subsection{Constructing $D_\alpha$}

Build $D_\alpha$ by adding to $\partial_\alpha$ the same correction terms found in \eqref{e:dtheta-transform}, but with the fixed parameter $\chi$ replaced by the coordinate $\theta$ itself, and the sign of the correction flipped:
\begin{equation}\label{e:Ddef}
   D_\alpha \equiv \partial_\alpha + i\,(C\gamma^\mu\theta)_\alpha\,\partial_\mu - i\,(\gamma^1\theta)_\alpha\,\partial_\xi\,.
\end{equation}
To check this is covariant, build $D'_\alpha$ the same way out of the primed variables, using $\partial'_\mu=\partial_\mu$, $\partial'_\xi=\partial_\xi$ \eqref{e:dxdxi-transform}, then substitute $\theta'=\theta+\chi$ and \eqref{e:dtheta-transform}:
\begin{equation}
\begin{aligned}
   D'_\alpha
   &= \partial'_\alpha + i\,(C\gamma^\mu\theta')_\alpha\,\partial_\mu - i\,(\gamma^1\theta')_\alpha\,\partial_\xi
   \\
   &= \Bigl[\partial_\alpha - i(C\gamma^\mu\chi)_\alpha\partial_\mu+i(\gamma^1\chi)_\alpha\partial_\xi\Bigr]
      + i(C\gamma^\mu\theta)_\alpha\partial_\mu+i(C\gamma^\mu\chi)_\alpha\partial_\mu
      - i(\gamma^1\theta)_\alpha\partial_\xi-i(\gamma^1\chi)_\alpha\partial_\xi
   \\
   &= \partial_\alpha+i(C\gamma^\mu\theta)_\alpha\partial_\mu-i(\gamma^1\theta)_\alpha\partial_\xi \;=\; D_\alpha\,,
\end{aligned}
\end{equation}
the $\chi$-dependent pieces cancelling exactly against the correction already present in $\partial'_\alpha$. So $D_\alpha$ is form-invariant --- the defining property of a covariant derivative. Using $C\gamma^0=-I$, $C\gamma^1=\sigma_3$, $\gamma^1=\sigma_1$ (Appendix~\ref{s:notation}), in components:
\begin{equation}\label{e:Dcomponents}
   D_1 = \partial_{\theta^1} - i\theta^1(\partial_0-\partial_1) - i\theta^2\partial_\xi\,,
   \qquad
   D_2 = \partial_{\theta^2} - i\theta^2(\partial_0+\partial_1) - i\theta^1\partial_\xi\,.
\end{equation}
Setting $\partial_\xi\to0$ recovers the un-extended covariant derivative used before the central charge was introduced.

\subsection{$D_\alpha$ reproduces the SUSY algebra}

As a check, direct computation (tracking the Grassmann signs exactly as in Appendix~\ref{s:grassmann}) gives
\begin{equation}\label{e:Dalgebra}
   D_1^2 = -i(\partial_0-\partial_1)\,,
   \qquad
   D_2^2 = -i(\partial_0+\partial_1)\,,
   \qquad
   \{D_1,D_2\} = -2i\,\partial_\xi\,.
\end{equation}
With the standard identification $P_\mu=-i\partial_\mu$, and $Z=-i\partial_\xi$ by the same pattern, \eqref{e:Dalgebra} is exactly $Q_1^2=P_0-P_1$, $Q_2^2=P_0+P_1$, $\{Q_1,Q_2\}=2Z$ --- the abstract algebra of \eqref{e:SUSYalgebra-components}, reproduced term for term by this differential-operator realization of $Q_\alpha$.

%
%
\section{The action}
\label{s:action}

Everything up to this point has used only properties of the extended SUSY algebra derived in \S\ref{s:superspace} --- no action or Lagrangian has been written down, yet it already fixed how the spectrum of any theory realizing this algebra must be organized, including the boson/fermion mass-degeneracy theorem of \S\ref{ss:supermultiplets}. We now give two concrete examples of a model with an action, built from the covariant derivative $D_\alpha$ of \S\ref{s:coDer} --- whose algebra \eqref{e:Dalgebra} was already checked to carry an honest representation of the SUSY algebra --- together with the specific dynamics ($\mathcal{W}$, $\gamma$) that the abstract algebra alone leaves undetermined. The two examples differ by whether the central charge $Z$ vanishes on the ground state or not: \S\ref{ss:vacuum} works out the $Z=0$ case; \S\ref{ss:Znotzero} then takes up $Z\neq0$.

%
%
\subsection{The ground state is the vacuum ($Z=0$)}
\label{ss:vacuum}

In this sector we assume that the ground state is the vacuum, that is $U(\tg) \ket{0} = \ket{0}$.  This means for the whole extended group, every generator annihilates $\ket{0}$: $P_\mu\ket{0}=Q_\alpha\ket{0}=Z\ket{0}=0$. But the algebra forces a consistency condition, since $\{Q_1,Q_2\}=2Z$, having both $Q_1\ket{0}=0$ and $Q_2\ket{0}=0$ requires $Z\ket{0}=0$. So a fully SUSY-invariant vacuum is automatically a $Z=0$ state, and we can ignore the $\xi$ variable.

As in ordinary field theory, write the action as $S[\Psi]=\int\rd^2x\,\rd^2\theta\,\mathcal{L}$, with the superspace Lagrangian
\begin{equation}\label{e:action-Z0}
   \mathcal{L} = \tfrac{\alpha}{2}\,\overline{D\Psi}\,(D\Psi) + i\gamma\,\mathcal{W}(\Psi) \,,
   \qquad
   \alpha,\gamma\in\mathbb{R}\,,
\end{equation}
where $D\Psi\equiv(D_1\Psi,D_2\Psi)^T$ and $\overline{D\Psi}\equiv(D\Psi)^TC$, the same bar used throughout for $\bar\theta\theta,\bar\psi\psi$. Since $D_1\Psi,D_2\Psi$ are each Grassmann-odd and so anticommute, $\overline{D\Psi}\,(D\Psi)=D_1\Psi\,D_2\Psi-D_2\Psi\,D_1\Psi=2\,D_1\Psi\,D_2\Psi$, so
\begin{equation}\label{e:L-Z0-simplified}
   \mathcal{L} = \alpha\,D_1\Psi\,D_2\Psi + i\gamma\,\mathcal{W}(\Psi)\,,
\end{equation}
the form used below. The real constants $\alpha,\gamma$ are not fixed by Hermiticity alone and are pinned down below: $\alpha=-\tfrac12$ by the boson/fermion kinetic sign, while $\gamma$ remains a free coupling.

It is convenient to collect the combination $\gamma/(2\alpha)$, which recurs throughout everything that follows, into a single rescaled superpotential,
\begin{equation}\label{e:Wdef}
   W(\phi) \equiv \frac{\gamma}{2\alpha}\,\mathcal{W}(\phi)\,,
\end{equation}
so that $W',W'',W'''$ always mean $\frac{\gamma}{2\alpha}\mathcal{W}',\frac{\gamma}{2\alpha}\mathcal{W}'',\frac{\gamma}{2\alpha}\mathcal{W}'''$. $W$ carries no content beyond $\mathcal{W}$, $\alpha$, and $\gamma$ --- it is bookkeeping, not a new coupling --- but it keeps the equations of motion below free of repeated fractions.

%
%
\subsubsection{The Super Lagrange equation}
\label{s:combinedEOM-Z0}

We now want to vary the action \eqref{e:action-Z0} to find the equation of motion for the superfield $\Psi$: varying the action directly gives a single combined equation of motion in one step, exactly as in ordinary field theory. The starting point is simply the definition of the first variation,
\begin{equation}\label{e:deltaS-def}
   \delta S = \int\rd^2x\,\rd^2\theta\,\delta\mathcal{L}\,,
\end{equation}
and everything below is just computing $\delta\mathcal{L}$, rewriting the result as $\delta\Psi$ times something, and then setting that something to zero.

Substitute $\Psi\to\Psi+\delta\Psi$ into the kinetic term of \eqref{e:L-Z0-simplified}. Since $D_1,D_2$ are linear, $D_\alpha(\Psi+\delta\Psi)=D_\alpha\Psi+D_\alpha\delta\Psi$, so
\begin{equation}\label{e:PsiPlusDeltaPsi}
   D_1(\Psi+\delta\Psi)\,D_2(\Psi+\delta\Psi)
   = \underbrace{D_1\Psi\,D_2\Psi}_{O(\delta\Psi^0)}
   + \underbrace{D_1\delta\Psi\,D_2\Psi + D_1\Psi\,D_2\delta\Psi}_{O(\delta\Psi^1)}
   + \underbrace{D_1\delta\Psi\,D_2\delta\Psi}_{O(\delta\Psi^2)}\,.
\end{equation}
The zeroth-order term is the unperturbed Lagrangian, and the quadratic term is dropped as higher order in the infinitesimal $\delta\Psi$ --- the same ``first variation'' logic as ordinary calculus of variations. What survives is the first-order piece, $\delta(D_1\Psi\,D_2\Psi)=D_1\delta\Psi\,D_2\Psi+D_1\Psi\,D_2\delta\Psi$; together with the ordinary Taylor expansion $\mathcal{W}(\Psi+\delta\Psi)=\mathcal{W}(\Psi)+\mathcal{W}'(\Psi)\,\delta\Psi+O(\delta\Psi^2)$ of the real function $\mathcal{W}$, this gives
\begin{equation}\label{e:deltaL-Z0}
   \delta\mathcal{L} = \alpha\bigl[D_1\delta\Psi\,D_2\Psi+D_1\Psi\,D_2\delta\Psi\bigr] + i\gamma\,\mathcal{W}'(\Psi)\,\delta\Psi\,.
\end{equation}
The $\mathcal{W}'(\Psi)\,\delta\Psi$ piece is already in the form we need, a single $\delta\Psi$ times something; the two kinetic cross-terms are not, since $\delta\Psi$ sits inside a $D_\alpha$ in each. Moving it out front is a job for integration by parts --- but with $\theta$, not $x$, as the variable being integrated over, so it is \emph{Grassmann} integration by parts: in place of the ordinary total-derivative identity $\int\rd x\,\partial_x(\text{anything})=0$, we need its superspace counterpart for $D_\alpha$.

This is possible because $\int d^2\theta\,\partial_\alpha(\text{anything})=0$ identically (differentiating always lowers the $\theta$-degree, so it can never reach the top component), and since $D_\alpha=\partial_\alpha+\theta\times(\text{spacetime derivative})$, it follows that $\int\rd^2x\,\rd^2\theta\,D_\alpha(X)=0$ for any superfield $X$: the leftover piece is a total $\partial_\mu$, which vanishes on integrating over $x$ provided that piece dies off at $x=\pm\infty$. Here $X$ always carries a factor of $\delta\Psi$, so this only requires $\delta\Psi\to0$ at the spatial boundary with $\Psi$'s own boundary values held fixed --- the same assumption used to drop boundary terms when deriving the ordinary Euler--Lagrange equations from $\delta S=0$. That holds regardless of what $\Psi$ itself does at $x=\pm\infty$, including on a kink background where $\phi$ approaches two different vacua (\S\ref{ss:Znotzero}): it is $\delta\Psi$, not $\Psi$, that is required to vanish at the boundary. So $D_\alpha$ can be integrated by parts inside $\int\rd^2x\,\rd^2\theta$, exactly like an ordinary $\partial_\mu$, with the graded Leibniz rule
\begin{equation}\label{e:gradedLeibniz}
   D_\alpha(AB)=(D_\alpha A)B+(-1)^{|A|}A(D_\alpha B)\,,
\end{equation}
where $|A|$ is the Grassmann parity of $A$: $|A|=0$ if $A$ is Grassmann-even (bosonic) and $|A|=1$ if $A$ is Grassmann-odd (fermionic), so $(-1)^{|A|}=+1,-1$ respectively.

The two kinetic cross-terms are now integrated by parts one at a time using \eqref{e:gradedLeibniz}. For the first, apply $D_1$ to the product $(\delta\Psi)(D_2\Psi)$, with $A=\delta\Psi$ even ($|A|=0$):
\begin{equation}\label{e:IBP-first}
   D_1\bigl[(\delta\Psi)(D_2\Psi)\bigr] = (D_1\delta\Psi)(D_2\Psi) + (\delta\Psi)(D_1D_2\Psi)\,.
\end{equation}
Integrating both sides over $\int\rd^2x\,\rd^2\theta$, the left side vanishes by the total-derivative fact established above, giving
\begin{equation}\label{e:IBP-first-result}
   \int D_1\delta\Psi\,D_2\Psi = -\int\delta\Psi\,D_1D_2\Psi\,.
\end{equation}
For the second, apply $D_2$ to the product $(D_1\Psi)(\delta\Psi)$, with $A=D_1\Psi$ \emph{odd} ($|A|=1$, since $D_1$ is an odd operator acting on the even $\Psi$):
\begin{equation}\label{e:IBP-second}
   D_2\bigl[(D_1\Psi)(\delta\Psi)\bigr] = (D_2D_1\Psi)(\delta\Psi) - (D_1\Psi)(D_2\delta\Psi)\,.
\end{equation}
Integrating and setting the left side to zero,
\begin{equation}\label{e:IBP-second-result}
   \int D_1\Psi\,D_2\delta\Psi = \int (D_2D_1\Psi)(\delta\Psi) = \int\delta\Psi\,D_2D_1\Psi\,,
\end{equation}
where the last step reorders $\delta\Psi$ to the left at no cost, since $D_2D_1\Psi$ (two odd $D$'s composed) and $\delta\Psi$ are both Grassmann-even and so commute freely.

Both kinetic cross-terms in \eqref{e:deltaL-Z0} now have $\delta\Psi$ out front, exactly like the $\mathcal{W}'(\Psi)\,\delta\Psi$ piece. Substituting \eqref{e:IBP-first-result} and \eqref{e:IBP-second-result} into \eqref{e:deltaS-def}, term by term,
\begin{equation}\label{e:deltaS-kinetic}
   \delta S = \int\delta\Psi\,\Bigl[-\alpha\,D_1D_2\Psi \;+\; \alpha\,D_2D_1\Psi \;+\; i\gamma\,\mathcal{W}'(\Psi)\Bigr]\,.
\end{equation}
A direct check, identical in method to \eqref{e:Dalgebra}, shows $\{D_1,D_2\}=0$ exactly for the $Z=0$ covariant derivative (the $\xi$-independent specialization of \eqref{e:Ddef}) --- it is \eqref{e:Dalgebra}'s $\{D_1,D_2\}=-2i\partial_\xi$ with the $\xi$-piece simply absent, so $D_2D_1\Psi=-D_1D_2\Psi$. Substituting this into the middle term only, $\alpha\,D_2D_1\Psi=-\alpha\,D_1D_2\Psi$, so the bracket becomes $-\alpha D_1D_2\Psi-\alpha D_1D_2\Psi+i\gamma\mathcal{W}'(\Psi)=-2\alpha\,D_1D_2\Psi+i\gamma\,\mathcal{W}'(\Psi)$. Setting $\delta S=0$ for arbitrary $\delta\Psi$ gives $D_1D_2\Psi=\frac{i\gamma}{2\alpha}\mathcal{W}'(\Psi)$; using the rescaled superpotential \eqref{e:Wdef}, this is the \textbf{Super Lagrange equation}
\begin{equation}\label{e:combinedEOM}
   \boxed{\;D_1D_2\Psi = i\,W'(\Psi)\;}\,,
\end{equation}
the superspace analogue of the ordinary Euler--Lagrange equation $\delta S/\delta\Psi=0$. Here $W'(\Psi)$ means $W'$ evaluated with $\Psi$ substituted in, then Taylor-expanded in $\theta$ the same way any superfield is --- not a new object, and consistent with the chain rule applied directly to the $i\gamma\,\mathcal{W}(\Psi)$ term in the action, which varies to $i\gamma\,\mathcal{W}'(\Psi)\,\delta\Psi=2i\alpha\,W'(\Psi)\,\delta\Psi$.

%
%
\subsubsection{Components of the Super Lagrange equation}
\label{s:components-Z0}

Casting \eqref{e:combinedEOM} into ordinary component fields needs $D_\alpha$ and $\Psi$ written out explicitly. From \eqref{e:Ddef}, the covariant derivative is
\begin{equation}\label{e:Ddef-Z0}
   D_\alpha = \partial_\alpha + i\,(C\gamma^\mu\theta)_\alpha\,\partial_\mu \>.
\end{equation}
and specializing \eqref{e:Phi-def} to $\xi$-independence (with $\Phi\to\Psi$, since here $\Psi$ is the dynamical field the action is built from), the function $\Psi(x,\theta)$ becomes
\begin{equation}\label{e:Phi-def-Z0}
\begin{aligned}
   \Psi(x,\theta)
   &=
   \phi(x) + \bar\theta\psi(x) + \tfrac{i}{2}\,\bar\theta\theta\,F(x) \,,
   \\
   &=
   \phi(x) + [\, \theta^1\psi^2(x)-\theta^2\psi^1(x) \,]
      + i\,\theta^1\theta^2\,F(x) \>,
\end{aligned}
\end{equation}
where $\phi(x)$ and $F(x)$ are real functions, and $\psi(x)$ is a real Majorana Grassmann spinor.

A direct computation from \eqref{e:Ddef-Z0} and \eqref{e:Phi-def-Z0}, discarding every term with a
repeated $\theta$, gives
\begin{equation}\label{e:D2Phi-Z0}
   D_2\Psi = -\psi^1 - i\theta^1F - i\theta^2(\partial_0+\partial_1)\phi
              + i\theta^1\theta^2(\partial_0+\partial_1)\psi^2 \>.
\end{equation}
Applying $D_1$ once more, by the same rule, gives $D_1D_2\Psi$ directly:
\begin{equation}\label{e:D1D2Psi}
   D_1D_2\Psi = -iF + i\theta^1(\partial_0-\partial_1)\psi^1 + i\theta^2(\partial_0+\partial_1)\psi^2 - \theta^1\theta^2(\partial_0^2-\partial_1^2)\phi\,,
\end{equation}
To get $W'(\Psi)$'s expansion in components, write $\Psi=\phi+\eta$ with $\eta\equiv\bar\theta\psi+\tfrac{i}{2}\bar\theta\theta\,F=\theta^1\psi^2-\theta^2\psi^1+i\theta^1\theta^2F$, and Taylor-expand $W(\Psi)=W(\phi+\eta)$ itself ($\eta^3=0$ identically, only two Grassmann generators):
\begin{equation}\label{e:WPhi-Z0}
   W(\Psi) = W(\phi) + W'(\phi)\,\eta + \tfrac12W''(\phi)\,\eta^2\,.
\end{equation}
The same expansion, applied to the function $W'$ in place of $W$, gives
\begin{equation}\label{e:WprimePhi-Z0}
   W'(\Psi) = W'(\phi) + W''(\phi)\,\eta + \tfrac12W'''(\phi)\,\eta^2\,.
\end{equation}
Since $\eta=\bar\theta\psi$ is purely linear in $\theta$, $\eta^2=(\bar\theta\psi)^2=-2\,\theta^1\theta^2\,\psi^1\psi^2$ (Appendix~\ref{s:grassmann}), so \eqref{e:WprimePhi-Z0} becomes, in full component form,
\begin{equation}\label{e:WprimePhi-components}
   W'(\Psi) = W'(\phi) + \theta^1W''(\phi)\psi^2 - \theta^2W''(\phi)\psi^1 + \theta^1\theta^2\bigl[iW''(\phi)F - W'''(\phi)\psi^1\psi^2\bigr]\,.
\end{equation}
Matching \eqref{e:D1D2Psi} against $i\,W'(\Psi)$ \eqref{e:WprimePhi-components} term by term in $\theta$ reproduces, all at once:
\begin{equation}\label{e:combinedEOM-components}
\begin{aligned}
   \theta^0:&\quad F = -W'(\phi)\,,
   \\
   \theta^1,\theta^2:&\quad (\partial_0\mp\partial_1)\psi^{1,2} = \pm W''(\phi)\,\psi^{2,1}\,,
   \\
   \theta^1\theta^2:&\quad (\partial_0^2-\partial_1^2)\phi = -W'(\phi)W''(\phi) + iW'''(\phi)\,\psi^1\psi^2\,.
\end{aligned}
\end{equation}
The first line is the auxiliary-field relation; the next two are the two fermion equations of motion; the last is the \emph{fully nonlinear} boson equation of motion, carrying a fermion-bilinear back-reaction term that a purely bosonic treatment of $F$ would miss. The single Super Lagrange equation \eqref{e:combinedEOM} thus already contains all of this at once.

%
%
\subsubsection{The mass case ($W'''=0$)}

For $\mathcal{W}(\phi)=\tfrac{m}{2}\phi^2$, so that $W(\phi)=\tfrac{\gamma}{2\alpha}\mathcal{W}(\phi)=\tfrac{\mu}{2}\phi^2$ with $\mu\equiv\gamma m/(2\alpha)$ (giving $W'=\mu\phi$, $W''=\mu$, $W'''=0$; $m$ real), the fermion-bilinear term in \eqref{e:combinedEOM-components} vanishes identically, and the components become
\begin{equation}\label{e:massEOM}
   F=-\mu\phi\,,
   \qquad
   (\partial_0-\partial_1)\psi^1=\mu\psi^2\,,
   \qquad
   (\partial_0+\partial_1)\psi^2=-\mu\psi^1\,,
   \qquad
   \bigl(\partial_0^2-\partial_1^2+\mu^2\bigr)\phi=0\,.
\end{equation}
The last equation is the ordinary Klein--Gordon equation with mass $|\mu|$. To see that the two fermion equations are the 2D Dirac equation, write $\psi=(\psi^1,\psi^2)^T$ and use the real gamma-matrix representation \eqref{e:gamma01}, $\gamma^0=i\sigma_2$, $\gamma^1=\sigma_1$:
\begin{equation}
   \gamma^0\partial_0\psi = \begin{pmatrix}\partial_0\psi^2\\-\partial_0\psi^1\end{pmatrix}\,,
   \qquad
   \gamma^1\partial_1\psi = \begin{pmatrix}\partial_1\psi^2\\ \partial_1\psi^1\end{pmatrix}\,,
\end{equation}
so the Dirac equation $\gamma^\mu\partial_\mu\psi+\mu\psi=0$ reads, component by component,
\begin{equation}
   \begin{pmatrix}(\partial_0+\partial_1)\psi^2+\mu\psi^1\\[2pt](\partial_1-\partial_0)\psi^1+\mu\psi^2\end{pmatrix} = \begin{pmatrix}0\\0\end{pmatrix}\,,
\end{equation}
i.e.\ $(\partial_0+\partial_1)\psi^2=-\mu\psi^1$ and $(\partial_0-\partial_1)\psi^1=\mu\psi^2$ --- exactly the two fermion equations in \eqref{e:massEOM}. So $\psi$ obeys
\begin{equation}\label{e:Diraceq-Z0}
   (\slashed\partial+\mu)\psi=0\,,
   \qquad
   \slashed\partial\equiv\gamma^\mu\partial_\mu\,,
\end{equation}
the 2D Dirac equation for the Majorana spinor $\psi$, with mass $|\mu|$ equal to the boson mass of the Klein--Gordon equation above --- the expected boson/fermion mass degeneracy of an unbroken SUSY multiplet. Since $W'''=0$, there is no fermion back-reaction on $\phi$ at this order: the mass case is exactly linear/free.

%
%
\subsubsection{The interacting case ($W'''\neq0$)}

Take $\mathcal{W}(\phi)=\tfrac m2\phi^2+\tfrac\lambda3\phi^3$, so
\begin{equation}\label{e:Wcubic}
   \mathcal{W}'=m\phi+\lambda\phi^2\,,
   \qquad
   \mathcal{W}''=m+2\lambda\phi\,,
   \qquad
   \mathcal{W}'''=2\lambda\,,
\end{equation}
a nonzero constant as soon as $\lambda\neq0$. The rescaled superpotential \eqref{e:Wdef} inherits this directly: with $\mu\equiv\gamma m/(2\alpha)$ as before and the new constant $W'''\equiv\frac{\gamma}{2\alpha}\mathcal{W}'''=\frac{\gamma\lambda}{\alpha}$,
\begin{equation}\label{e:WcubicRescaled}
   W'(\phi) = \mu\phi+\tfrac12W'''\phi^2\,,
   \qquad
   W''(\phi) = \mu+W'''\phi\,.
\end{equation}
Substituting into \eqref{e:combinedEOM-components} gives
\begin{equation}\label{e:FWcubic}
   F = -W'(\phi) = -\Bigl(\mu\phi+\tfrac12W'''\phi^2\Bigr)\,,
\end{equation}
\begin{equation}\label{e:fermionWcubic}
   (\partial_0\mp\partial_1)\psi^{1,2} = \pm W''(\phi)\,\psi^{2,1} = \pm\bigl(\mu+W'''\phi\bigr)\psi^{2,1}\,,
\end{equation}
\begin{equation}\label{e:bosonWcubic}
   (\partial_0^2-\partial_1^2)\phi + \mu^2\phi = -\tfrac32\mu W'''\,\phi^2-\tfrac12(W''')^2\,\phi^3 + iW'''\,\psi^1\psi^2\,.
\end{equation}
Three concrete effects of turning on $\lambda$, each visible explicitly rather than hidden inside $W',W'',W'''$: the fermion mass in \eqref{e:fermionWcubic} becomes field-dependent, $\mu\to\mu+W'''\phi$, so the fermion sees a position-dependent effective mass once $\phi$ is excited; the boson equation \eqref{e:bosonWcubic} picks up the quadratic and cubic self-interaction terms $-V'(\phi)$ for the quartic potential
\begin{equation}\label{e:Vdef}
   V(\phi) \equiv \tfrac12\,W'(\phi)^2\,,
\end{equation}
whose derivative $V'(\phi)=W'(\phi)W''(\phi)$ is exactly the \eqref{e:combinedEOM-components} boson source before the $\mu^2\phi$ piece is moved to the left-hand side; and it also picks up the fermion back-reaction term $iW'''\psi^1\psi^2$, with coefficient fixed by the same $\lambda$ that sets the cubic term in $\mathcal{W}$ --- not an independent coupling. Setting $\lambda=0$ recovers \eqref{e:massEOM} exactly, confirming $W'''=0$ is special to the purely quadratic case.

Since $V\equiv\tfrac12W'^2$ is manifestly non-negative, its minimum value is exactly $0$, attained wherever $W'(\phi)=0$: at $\phi=0$ for any $\lambda$, and, once $\lambda\neq0$, also at $\phi=-m/\lambda$. Since $F=-W'(\phi)$ \eqref{e:Wdef}, both are honest $F=0$ vacua: turning on the cubic coupling does not break supersymmetry, it only produces a second vacuum exactly degenerate with the first, the two endpoints between which the kink of \S\ref{ss:Znotzero} interpolates.

%
%
\subsection{The ground state is \emph{not} the vacuum ($Z\neq0$)}
\label{ss:Znotzero}

In this case, the action is given by
\begin{equation}\label{e:z1.action}
   S[\Psi] = \int \rd \ts \, \mathcal{L}[\Psi],
   \quad
   \rd \ts \equiv (\rd x,\rd \theta, \rd \xi) \>,
\end{equation}
where now the Lagrangian is given by:
\begin{equation}\label{e:z1.calL}
   \mathcal{L}[\Psi] = \alpha\,( D_1\Psi) \, ( D_2 \Psi) + i\gamma\,\mathcal{W}(\Psi) 
\end{equation}
However in this case the wave function now depends on $(x,\theta,\xi)$, and is of the form:
\begin{equation}\label{e:z1.Phi-def-Z0}
\begin{aligned}
   \Psi(x,\theta,\xi)
   &=
   \phi(x,\xi) + \bar\theta\,\psi(x,\xi) + \tfrac{i}{2}\,\bar\theta\theta\,F(x,\xi) \,,
   \\
   &=
   \phi(x,\xi) + [\, \theta^1\psi^2(x,\xi)-\theta^2\psi^1(x,\xi) \,]
      + i\,\theta^1\theta^2\,F(x,\xi) \>,
\end{aligned}
\end{equation}
where $\phi(x,\xi)$ and $F(x,\xi)$ are real functions, and $\psi(x,\xi)$ is a real Majorana Grassmann spinor.  In this case, the covariant derivative \eqref{e:Ddef}:
\begin{equation}\label{e:z1.Ddef}
   D_\alpha = \partial_\alpha + i\,(C\gamma^\mu\theta)_\alpha\,\partial_\mu - i\,(\gamma^1\theta)_\alpha\,\partial_\xi\,,
\end{equation}
which now contain derivatives with respect to $\xi$.  The same trick used in the $Z=0$ case in Eq~\eqref{e:IBP-second} applies here, namely:
\begin{equation}\label{e:z1.IBP-second}
   D_2\bigl[(D_1\Psi)(\delta\Psi)\bigr] = (D_2 D_1\Psi)(\delta\Psi) - (D_1\Psi)(D_2\delta\Psi)\,.
\end{equation}
Now integrating this equation over $\rd \ts \equiv (\rd x,\rd \theta, \rd \xi)$ gives:
\begin{equation}\label{e:z1.ByParts}
   \int \rd \ts \, (D_1\Psi)(D_2\delta\Psi)
   =
   \int \rd \ts \, (D_2 D_1\Psi)(\delta\Psi)
   -
   \int \rd \ts \, D_2\bigl[(D_1\Psi)(\delta\Psi)\bigr] \>.
\end{equation}
The last term is a surface term; Appendix~\ref{s:IBPextended} proves the general integration-by-parts identity this relies on and evaluates it explicitly for this $X=(D_1\Psi)(\delta\Psi)$ (Corollary~\ref{cor:IBPextended}). Besides the ordinary $x$-boundary piece, already familiar from \S\ref{ss:vacuum} and vanishing there the same way, it now also carries a genuinely new $\xi$-boundary piece,
\begin{equation}\label{e:z1.surfaceterm-main}
   \int\rd x\,\Bigl[\bigl(F-\partial_\xi\phi\bigr)\delta\phi - i\,\psi^2\,\delta\psi^1\Bigr]_{\xi\to-\infty}^{\xi\to+\infty}\,,
\end{equation}
which vanishes only once $\xi$ is given a definite global structure (compactified on a circle, or assigned a falloff condition) --- the loose end flagged in \S\ref{ss:vacuum} when $\xi$ was first introduced.

\paragraph{The extended off-shell Lagrangian.} Carrying out $\int\rd^2\theta\,\mathcal{L}[\Psi]$ directly in components (keeping $\partial_\xi\phi,\partial_\xi\psi\neq0$ throughout) gives
\begin{equation}\label{e:z1.Loffshell}
\begin{aligned}
   \mathcal{L}\big|_{\theta^1\theta^2} = {}& -\alpha(\partial_0\phi)^2+\alpha(\partial_1\phi)^2-\alpha F^2-\gamma\mathcal{W}'(\phi)F \\
   &+ i\alpha\bigl[\psi^1(\partial_0-\partial_1)\psi^1+\psi^2(\partial_0+\partial_1)\psi^2\bigr] - i\gamma\mathcal{W}''(\phi)\psi^1\psi^2 \\
   &+ \alpha(\partial_\xi\phi)^2 + i\alpha\bigl(\psi^1\partial_\xi\psi^2+\psi^2\partial_\xi\psi^1\bigr)\,,
\end{aligned}
\end{equation}
exactly the familiar $Z=0$ off-shell Lagrangian (first two lines, reproducing \eqref{e:combinedEOM-components}'s equations of motion when $\partial_\xi\to0$) plus two new terms (last line): a $\xi$-gradient term for $\phi$ and a matching $\xi$-derivative cross term for the fermions --- exactly what treating $\xi$ as a third coordinate would lead one to expect.

\paragraph{Varying directly in components.} With $S=\int\rd x\,\rd\xi\,\mathcal{L}|_{\theta^1\theta^2}$, the ordinary $x^0,x^1$ boundary terms are the usual ones (vanishing as before); integrating by parts in $\xi$ as well now picks up three new surface terms, one for each field carrying a $\xi$-derivative. $F$ has none --- its equation stays purely algebraic, $F=-W'(\phi)$ \eqref{e:Wdef}, unchanged. Varying $\phi$:
\begin{equation}\label{e:z1.bosonEOM}
   \bigl(\partial_0^2-\partial_1^2-\partial_\xi^2\bigr)\phi = -W'(\phi)W''(\phi) + iW'''(\phi)\,\psi^1\psi^2\,,
\end{equation}
the full nonlinear boson equation with a genuine $\partial_\xi^2\phi$ term now included, together with the surface term $\bigl[2\alpha\,\partial_\xi\phi\,\delta\phi\bigr]_{\xi\to-\infty}^{\xi\to+\infty}$. Varying $\psi^1,\psi^2$ gives
\begin{equation}\label{e:z1.fermionEOM}
   (\partial_0-\partial_1)\psi^1 = -\partial_\xi\psi^2 + W''(\phi)\psi^2\,,
   \qquad
   (\partial_0+\partial_1)\psi^2 = -\partial_\xi\psi^1 - W''(\phi)\psi^1\,,
\end{equation}
together with surface terms $\bigl[i\alpha\,\psi^2\,\delta\psi^1\bigr]_{\xi\to-\infty}^{\xi\to+\infty}$ and $\bigl[i\alpha\,\psi^1\,\delta\psi^2\bigr]_{\xi\to-\infty}^{\xi\to+\infty}$ respectively. Collecting all three fields,
\begin{equation}\label{e:z1.allsurface}
   \int\rd x\,\Bigl[2\alpha\,\partial_\xi\phi\,\delta\phi + i\alpha\,\psi^2\,\delta\psi^1 + i\alpha\,\psi^1\,\delta\psi^2\Bigr]_{\xi\to-\infty}^{\xi\to+\infty}\,,
\end{equation}
three independently-vanishing boundary conditions, each again requiring $\xi$ to be either compactified or assigned a falloff condition.

\paragraph{Compactifying $\xi$: Fourier modes and the Kaluza--Klein tower.} Take $\xi$ compact, $\xi\in[0,L)$, periodic. Every surface term above, being a total $\xi$-derivative, then integrates to zero over one period automatically --- Eqs.~\eqref{e:z1.surfaceterm-main} and \eqref{e:z1.allsurface} vanish identically, with no separate falloff assumption required. Expand each field in Fourier modes,
\begin{equation}\label{e:z1.Fourier}
   \phi(x,\xi) = \sum_{n=-\infty}^{\infty}\phi_n(x)\,e^{iM_n\xi}\,,
   \qquad
   M_n \equiv \frac{2\pi n}{L}\,,
   \qquad n\in\mathbb{Z}\,,
\end{equation}
and likewise for $F,\psi^1,\psi^2$; reality of the $\xi$-dependent fields requires $\phi_n^*=\phi_{-n}$ (and similarly for $F_n,\psi^1_n,\psi^2_n$).

The quantization $M_n=2\pi n/L$ so far is pure Fourier analysis, fixed by periodicity on the circle alone and independent of any physics. What turns it into a statement about the central charge is the already-established identity
\begin{equation}
   \mathcal{Z}\Psi(x,\theta,\xi) \equiv \braket{x,\theta,\xi|Z|\Psi} = -i\,\partial_\xi\Psi(x,\theta,\xi)
\end{equation}
\eqref{e:PQZ-diff-ops}: if $\ket{\Psi}$ is an eigenstate of the abstract operator $Z$ with eigenvalue $z$, projecting $Z\ket{\Psi}=z\ket{\Psi}$ onto $\bra{x,\theta,\xi}$ gives $-i\partial_\xi\Psi=z\Psi$, so \emph{any} solution of this equation for constant $z$ is automatically a $Z$-eigenstate with that eigenvalue. Each mode solves exactly this equation,
\begin{equation}\label{e:z1.Zeigen}
   -i\,\partial_\xi\Bigl[\phi_n(x)\,e^{iM_n\xi}\Bigr] = M_n\,\phi_n(x)\,e^{iM_n\xi}\,,
\end{equation}
so mode $n$ is genuinely, provably a $Z$-eigenstate with eigenvalue $M_n$ --- not assumed, derived. Writing that eigenvalue as $Z_n\equiv M_n$ (the ordinary habit of naming an eigenvalue after its operator) therefore reports a fact rather than making a definition: the mode number $n$ labels the value of the \emph{same} central charge appearing throughout this paper, quantized in units of $2\pi/L$.

For the mass case, $\mathcal{W}(\phi)=\tfrac m2\phi^2$, $W''=\mu$ is constant, so \eqref{e:z1.allsurface}'s building blocks --- Eqs.~\eqref{e:z1.bosonEOM}--\eqref{e:z1.fermionEOM} together with $F=-W'(\phi)$ --- decouple mode by mode with no mixing between different $n$. Substituting \eqref{e:z1.Fourier} and $\partial_\xi\to iM_n$,
\begin{equation}\label{e:z1.Fn}
   F_n = -\mu\,\phi_n\,,
   \qquad
   \mu \equiv \frac{\gamma m}{2\alpha}\,,
\end{equation}
\begin{equation}\label{e:z1.fermionmodes}
   (\partial_0-\partial_1)\psi^1_n = (\mu-iM_n)\,\psi^2_n\,,
   \qquad
   (\partial_0+\partial_1)\psi^2_n = -(\mu+iM_n)\,\psi^1_n\,,
\end{equation}
\begin{equation}\label{e:z1.bosonmodes}
   \bigl(\partial_0^2-\partial_1^2+\mu^2+M_n^2\bigr)\phi_n = 0\,.
\end{equation}
Applying $(\partial_0+\partial_1)$ to the first equation in \eqref{e:z1.fermionmodes} and substituting the second confirms the fermions satisfy the identical mass-shell condition,
\begin{equation}\label{e:z1.fermionKG}
   \bigl(\partial_0^2-\partial_1^2+\mu^2+M_n^2\bigr)\psi^1_n = 0\,,
\end{equation}
so boson/fermion mass degeneracy survives mode by mode. Each level of the tower is thus an ordinary massive multiplet with
\begin{equation}\label{e:z1.KKmass}
   M_{\rm eff}^2 = \mu^2+M_n^2 = \mu^2+Z_n^2\,,
\end{equation}
the mass picking up a contribution from the square of its own central charge --- the generic (non-BPS) mass formula for a Kaluza--Klein tower.

This had to work out so cleanly: $Z$ is central, so by Schur's lemma it takes a single fixed value on any irreducible representation of the algebra, and $[Q_\alpha,Z]=0$ means the supercharges $Q_1,Q_2$ preserve that value exactly --- even though $\mathcal{Q}_\alpha$'s differential form \eqref{e:PQZ-diff-ops} does act through $\partial_\xi$, on a single Fourier mode that just multiplies by the constant $iM_n$, so $Q_1,Q_2$ never mix different $n$ into one another. So the compactified theory is not one multiplet carrying one $Z$; it is an infinite tower of \emph{separate} multiplets, one per $n$, each its own representation of the algebra with its own fixed $Z_n=M_n$. This representation-theoretic statement is exact and holds regardless of interactions --- $n=0$, with $Z_0=0$, is exactly the sector where the topological kink below lives undisturbed by any of it.

What interactions \emph{do} couple is the classical dynamics of a configuration exciting several modes at once, not the representation content itself. For the cubic superpotential \eqref{e:Wcubic}, $W''=\mu+W'''\phi$ is no longer constant, and $F_n=-\bigl[\mu\,\phi_n+\tfrac12W'''\sum_k\phi_k\phi_{n-k}\bigr]$ sources each mode's equation of motion from products of the others, with $n_1+n_2=n_3$ acting as a selection rule on which modes talk to each other. This Fourier decomposition only closes this cleanly, mode by mode, for the linear (mass) case; the interacting dynamics is not worked out here, only flagged.

\paragraph{Kink solutions live entirely in the $n=0$ mode.} Truncating to $\phi_n=\psi^a_n=F_n=0$ for every $n\neq0$ is self-consistent even in the interacting case: the convolution $F_n=-\bigl[\mu\,\phi_n+\tfrac12W'''\sum_k\phi_k\phi_{n-k}\bigr]$ needs \emph{both} $\phi_k$ and $\phi_{n-k}$ nonzero to contribute, so with only $\phi_0\neq0$ every $F_{n\neq0}$ stays exactly zero, and $F_0=-W'(\phi_0)$ is exactly the original, un-extended auxiliary-field relation (likewise for the boson and fermion equations, since every $\xi$-derivative of a $\xi$-independent field vanishes identically). So the $n=0$ sector decouples completely from the Kaluza--Klein tower and obeys exactly the equations of \S\ref{ss:vacuum} --- a kink lives there untouched by anything built in this subsection. This also means $Z$ has two logically independent origins here: the \emph{topological} charge of an $x$-space kink, carried entirely by the $n=0$ mode and needing no $\xi$-dependence at all (worked out below); and the \emph{Kaluza--Klein} charge $Z_n=M_n$ of an excited mode \eqref{e:z1.KKmass}, needing no $x$-space topology at all. A generic configuration can carry both at once, independently.

\paragraph{The kink equations.} For a static ($\partial_0\phi_0=0$), purely bosonic ($\psi=0$) profile, the $n=0$ boson equation reduces to $\partial_1^2\phi_0=W'(\phi_0)W''(\phi_0)=V'(\phi_0)$ (recall $V\equiv\tfrac12W'^2$, \eqref{e:Vdef}), exactly the Euler--Lagrange equation of the static energy functional $E[\phi_0]=\int\rd x\,\bigl[\tfrac12(\partial_1\phi_0)^2+V(\phi_0)\bigr]$. Complete the square:
\begin{equation}\label{e:z1.CompleteSquare}
   E[\phi_0] = \int\rd x\,\tfrac12\bigl(\partial_1\phi_0\mp W'(\phi_0)\bigr)^2 \;\pm\;\int\rd x\,\partial_1\phi_0\,W'(\phi_0)\,.
\end{equation}
The second term is a total derivative, $\int\rd x\,\partial_xW(\phi_0)=W(\phi_0(+\infty))-W(\phi_0(-\infty))\equiv\Delta W$. Since the first term is a perfect square and therefore non-negative, $E[\phi_0]\geq\pm\Delta W$ holds for \emph{either} choice of sign and for \emph{any} static $\phi_0$ sharing these boundary conditions --- not just the energy-minimizing one. Combining both signs and identifying the static energy with the mass $M$,
\begin{equation}\label{e:z1.BPSbound}
   \boxed{\;M\geq|Z|\;}\,,
   \qquad
   Z \equiv \Delta W\,,
\end{equation}
the BPS bound, with $Z$ the topological charge announced above. The name records exactly this argument: Bogomolny~\cite{Bogomolny1976} and, independently, Prasad and Sommerfield~\cite{PrasadSommerfield1975} first showed that completing the square this way turns an energy bound for classical monopoles into an exact statement, saturated precisely on the solutions of the resulting first-order equation rather than the full second-order one; Witten and Olive~\cite{WittenOlive1978} later identified that classical bound with the exact SUSY-algebra statement $M\geq|Z|$ derived here, and a state saturating it is accordingly called a BPS state. Equality holds exactly when the perfect-square term vanishes identically, i.e.\
\begin{equation}\label{e:z1.Bogomolny}
   \boxed{\;\partial_1\phi_0 = \pm\,W'(\phi_0)\;}\,,
\end{equation}
the Bogomolny equation for the kink --- the saturated case of the bound just derived, with minimum energy $E_{\min}=|\Delta W|=|Z|$.

Equation \eqref{e:z1.Bogomolny} is first-order and autonomous, so solving it is an initial-value problem: its fixed points are the zeros of $\mathcal{W}'$, and a kink is a trajectory running between two of them. For the cubic superpotential \eqref{e:Wcubic}, with $\mathcal{W}'=\phi_0(m+\lambda\phi_0)$, the two vacua are $\phi_0=0$ and $\phi_0=-m/\lambda$, and \eqref{e:z1.Bogomolny} becomes, with $k\equiv\pm\gamma/(2\alpha)$, the logistic equation $d\phi_0/dx=k\,\phi_0(m+\lambda\phi_0)$, solved exactly by
\begin{equation}\label{e:z1.kinkprofile}
   \boxed{\;\phi_0(x) = -\frac{m}{2\lambda}\Bigl[1+\tanh\Bigl(\tfrac{mk}{2}(x-x_0)\Bigr)\Bigr]\;}\,,
\end{equation}
interpolating smoothly between the two vacua, with $x_0$ an arbitrary modulus and $2/|mk|$ the kink's width; Figure~\ref{fig:kink} plots it for $m=\lambda=\gamma=1$, $\alpha=-\tfrac12$ ($k=1$), $x_0=0$.

\begin{figure}[htbp]
   \centering
   \includegraphics[width=0.6\linewidth]{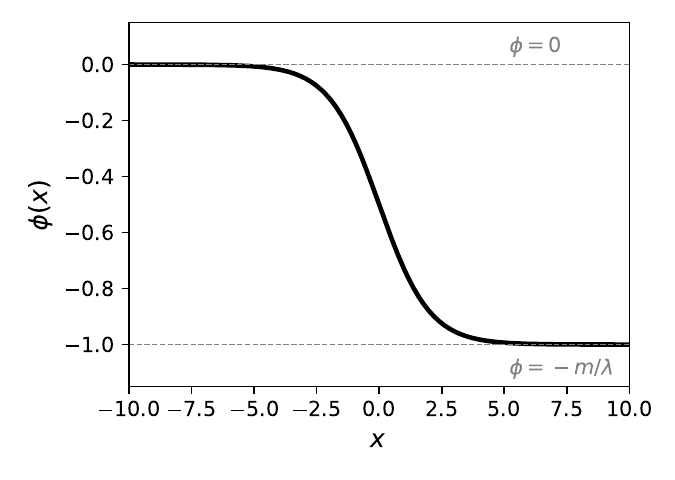}
   \caption{The kink profile \eqref{e:z1.kinkprofile}, living entirely in the $n=0$ mode, for $m=\lambda=\gamma=1$, $\alpha=-\tfrac12$, $x_0=0$, interpolating between the two vacua $\phi_0=0$ and $\phi_0=-m/\lambda=-1$.}
   \label{fig:kink}
\end{figure}

Fermion fluctuations around this background --- in particular, whether the sign change in $W''(\phi_0)$ between the two vacua produces a normalizable zero-energy bound state, as in the classic Jackiw--Rebbi construction~\cite{JackiwRebbi1976} --- are not addressed here.

%
%
\section{Conclusion}
\label{s:conclusion}

We built the 1+1D $\mathcal{N}{=}2$ Wess--Zumino model with a central extension from the top down: algebra first, dynamics second. \S\ref{s:superspace} derived the extended supertranslation group and its algebra from Wigner's theorem alone, with no action in sight, already fixing the Casimir invariants \eqref{e:Casimirs} and the boson/fermion mass-degeneracy theorem for any unitary representation; \S\ref{s:coDer} then built the covariant derivative $D_\alpha$ carrying that algebra. \S\ref{ss:vacuum} showed a single superspace equation, the Super Lagrange equation \eqref{e:combinedEOM}, contains the auxiliary-field relation and the full nonlinear fermion and boson equations of motion at once, for both the mass case and an interacting cubic superpotential. \S\ref{ss:Znotzero} then took $Z\neq0$ seriously by keeping the central extension's own coordinate $\xi$ explicit: this produces a new surface term and an explicit $\mathcal{Z}\Psi$ correction to the equations of motion (with the auxiliary-field relation unexpectedly protected from it), and compactifying $\xi$ turns that surface term into a Kaluza--Klein tower of separate supermultiplets, one per quantized eigenvalue $Z_n=2\pi n/L$, each with $M^2=\mu^2+Z_n^2$. The kink survives this construction untouched in the $n=0$ mode, saturating the general BPS bound $M\geq|Z|$ derived independently by completing the square in its static energy.

What is left open: the interacting superpotential's mode-mixing dynamics on the Kaluza--Klein tower (\S\ref{ss:Znotzero}) were set up but not solved. If the pedagogical case made in the Introduction is right, though, the value of this exercise was never in reaching a new result --- it was in seeing how much of the structure (mass degeneracy, the BPS bound, the shape of a central charge) was already fixed before a single component field was ever written down.

As a speculative aside, \S\ref{ss:Znotzero}'s construction has a particular shape: a compact extra coordinate, a resulting tower of states labeled by the conserved quantum number $Z_n$, and a distinguished zero mode that carries none of that quantum number itself but is where the topological kink lives. That is the same shape as the Kaluza--Klein dark matter scenario considered in models of Universal Extra Dimensions~\cite{AppelquistChengDobrescu2001}, where a discrete ``KK-parity'' stabilizes the lightest Kaluza--Klein particle as a dark matter candidate~\cite{ServantTait2003}, playing a role analogous to $R$-parity and the lightest supersymmetric particle in more conventional SUSY dark matter scenarios. The analogy should not be pushed far: $\xi$ here is bookkeeping for a central extension of the SUSY algebra, not a literal extra spatial dimension, and nothing in this note addresses whether any such stability mechanism survives once interactions are included. But the parallel is close enough to be worth noting as one of the more concrete ways this kind of superspace construction touches an open problem in physics.

\begin{acknowledgments}
The author gratefully acknowledges extensive assistance from Claude Sonnet 5~\cite{ClaudeSonnet5}, used throughout the development of this note both to work through and check derivations --- including several places where an initial calculation on either side contained a sign or algebra error later caught by cross-checking --- and to draft and revise portions of the expository text. All results and text were reviewed by the author, who takes full responsibility for the content.
\end{acknowledgments}

\appendix
%
%
\section{Notation}
\label{s:notation}

The extended set of Pauli matrices are defined by
\begin{equation}\label{e:PM.Pauli}
   \sigma_0 = I = \begin{pmatrix} 1 & 0 \\ 0 & 1 \end{pmatrix}
   \,,\qquad
   \sigma_1 = \begin{pmatrix} 0 & 1 \\ 1 & 0 \end{pmatrix}
   \,,\qquad
   \sigma_2 = \begin{pmatrix} 0 & -i \\ i & 0 \end{pmatrix}
   \,,\qquad
   \sigma_3 = \begin{pmatrix} 1 & 0 \\ 0 & -1 \end{pmatrix} \,.
\end{equation}
In 1+1D with a signature $(-,+)$, we choose a real representation of the $\gamma$-matrices as:
\begin{equation}\label{e:gamma01}
   \gamma^0=i\sigma_2=\begin{pmatrix}0&1\\-1&0\end{pmatrix}\>,\qquad
   \gamma^1=\sigma_1=\begin{pmatrix}0&1\\1&0\end{pmatrix}\>,
\end{equation}
where $\sigma_i$ are the Pauli matrices.  The gamma matrices then obey the Clifford algebra $\acomm{\gamma^\mu}{\gamma^\nu}=2\eta^{\mu\nu}$.  We then find that
\begin{equation}\label{e:PM}
   \gamma_0 = -\gamma^0\,,\quad \gamma_1 = \gamma^1 \>.
\end{equation}
We choose for the charge-conjugate operator:
\begin{equation}\label{e:Cdefinition}
    C = i\sigma_2 = \gamma^0 = \begin{pmatrix} 0 & 1 \\ -1 & 0 \end{pmatrix}\,,
    \qquad
    C^{-1} = - C \,.
\end{equation}
Then $C$ obeys the relation: $C \gamma^{\mu} C^{-1} = - (\gamma^{\mu})^T$, and has the properties:
\begin{equation}\label{e:Cproperties}
   C^T = -C\,,\quad
   C^2 = - \sigma_0 = - I\,,\quad
   C^{T}C = \sigma_0 = I\,,\quad
    C \gamma^0 = -\sigma_0 = - I\,, \quad C \gamma^1 = \sigma_3 \>.
\end{equation}

%
%
\section{Grassmann calculus}
\label{s:grassmann}

\begin{defn}[Grassmann variables]
A set of $N$ quantities $\theta_i$, $i = 1,2,\dotsc,N$ are Grassmann variables if they obey the anti-commutation relations:
\begin{equation}\label{ssym.e:thetarules}
   \acomm{\theta_i}{\theta_j} = 0 \>.
\end{equation}
Grassmann variables \emph{commute} with all other classical variables.
\end{defn}

\begin{defn}[differentiation]
Derivatives of Grassmann variables are taken to be \emph{left}-acting
and anti-commute.
\begin{equation}\label{ssym.e:diffrules}
   \acomm{\partial_i}{\theta_j}
   =
   \delta_{ij} \>,
   \qquad
   \acomm{\partial_i}{\partial_j}
   =
   0 \>,
   \qquad
   \partial_i
   \equiv
   \frac{\partial}{\partial \theta_i} \>.
\end{equation}
\end{defn}
This means, for example, that:
\begin{equation}
   \partial_i ( \, \theta_j \theta_k \, )
   =
   \delta_{ij} \, \theta_k - \delta_{ik} \, \theta_j \>.
\end{equation}

Integration of Grassmann variables has some unusual properties, which are given in the next two definitions.
\begin{defn}[integration]
The differential obeys the following rules:
\begin{equation}\label{ssym.e:differentialrules}
   \acomm{\rd\theta_i}{\theta_j}
   =
   0 \>,
   \qquad
   \acomm{\rd\theta_i}{\rd\theta_j}
   =
   0 \>,
   \qquad
   \acomm{\rd\theta_i}{\partial_j}
   =
   \delta_{ij} \>.
\end{equation}
Integrals are defined by:
\begin{equation}\label{ssym.e:intrules}
   \int \rd \theta = 0 \>,
   \qquad
   \int \rd \theta \, \theta = 1 \>.
\end{equation}
\end{defn}

The integration rules mean that a Grassmann Dirac $\delta$-function
can be defined by:
\begin{equation}\label{ssym.e:thetadeltaf}
   \delta(\theta) = \theta \>.
\end{equation}

\subsection{Complex Grassmann variables}

Up until now, we have been considering real Grassmann variables.  If $\theta_1$ and $\theta_2$ are two real Grassmann variables, complex Grassmann variables can be defined as follows:
\begin{alignat}{2}\label{ssym.e:thetacomplexdef}
   \theta
   &=
   ( \, \theta_1 + i \theta_2 \, ) / \sqrt{2} \>,
   & \qquad
   \theta_1
   &=
   ( \, \theta + \theta^{\ast} \, ) / \sqrt{2} \>,
   \\
   \theta^{\ast}
   &=
   ( \, \theta_1 - i \theta_2 \, ) / \sqrt{2} \>,
   & \qquad
   \theta_2
   &=
   ( \, \theta - \theta^{\ast} \, ) / i\sqrt{2} \>.
\end{alignat}
The complex variables satisfy the algebra:
\begin{equation}
   \acomm{\theta}{\theta}
   =
   \acomm{\theta^{\ast}}{\theta^{\ast}}
   =
   \acomm{\theta}{\theta^{\ast}}
   =
   0 \>.
\end{equation}
For the derivatives, we have:
\begin{alignat}{2}\label{ssym.e:thetaderivdef}
   \partial_\theta
   &=
   ( \, \partial_1 - i \partial_2 \, ) / \sqrt{2} \>,
   & \qquad
   \partial_1
   &=
   ( \, \partial_\theta + \partial_\theta^{\ast} \, ) / \sqrt{2} \>,
   \\
   \partial_\theta^{\ast}
   &=
   ( \, \partial_1 + i \partial_2 \, ) / \sqrt{2} \>,
   & \qquad
   \partial_2
   &=
   i ( \, \partial_\theta - \partial_\theta^{\ast} \, ) / \sqrt{2} \>.
\end{alignat}
The corresponding differentials transform the same way as $\theta$ and $\theta^{\ast}$ themselves,
\begin{equation}\label{ssym.e:thetacomplexdiffdef}
   \rd\theta
   =
   ( \, \rd\theta_1 + i \, \rd\theta_2 \, ) / \sqrt{2} \>,
   \qquad
   \rd\theta^{\ast}
   =
   ( \, \rd\theta_1 - i \, \rd\theta_2 \, ) / \sqrt{2} \>,
\end{equation}
so that, exactly as in Eq.~\eqref{ssym.e:differentialrules} for the real generators,
\begin{equation}
\begin{split}
   \acomm{\rd\theta}{\theta}
   =
   \acomm{\rd\theta}{\theta^{\ast}}
   =
   \acomm{\rd\theta^{\ast}}{\theta}
   =
   \acomm{\rd\theta^{\ast}}{\theta^{\ast}}
   &=
   0 \>,
   \\
   \acomm{\rd\theta}{\rd\theta}
   =
   \acomm{\rd\theta^{\ast}}{\rd\theta^{\ast}}
   =
   \acomm{\rd\theta}{\rd\theta^{\ast}}
   &=
   0 \>.
\end{split}
\end{equation}
Integrals over complex Grassmann variables are given by:
\begin{equation}
   \int \rd\theta =  \int \rd\theta^{\ast} = 0 \>,
   \qquad
   \int \rd\theta \, \theta
   =
   \int \rd\theta^{\ast} \, \theta^{\ast} = 1 \>,
\end{equation}
exactly Eq.~\eqref{ssym.e:intrules}'s single-variable rule, applied directly to $\theta,\theta^\ast$ as independent Grassmann generators in their own right. Also,
\begin{equation}
   \rd^2 \theta
   =
   \rd \theta \, \rd \theta^{\ast}
   =
   -i \, \rd\theta_1 \, \rd\theta_2 \>,
\end{equation}
using $\rd\theta_i^2=0$ and $\rd\theta_2\,\rd\theta_1=-\rd\theta_1\,\rd\theta_2$
from Eq.~\eqref{ssym.e:thetacomplexdiffdef}; equivalently, this is the determinant of the linear map taking $(\theta_1,\theta_2)$ to $(\theta,\theta^{\ast})$.
The complex conjugate of two Grassmann variables is defined by:
\begin{equation}\label{ssym.e:thetacomplexrule}
   [ \, \theta_1 \theta_2 \, ]^{\ast}
   =
   \theta_2^{\ast} \, \theta_1^{\ast} \>,
\end{equation}
which is similar to the Hermitian adjoint operation for matrices.

\subsection{Majorana Grassmann variables}

The two-component Majorana Grassmann spinors are defined by
\begin{equation}\label{e:GrassmannSpinorsDefs}
   \theta = \begin{pmatrix} \theta^1 \\ \theta^2 \end{pmatrix}\,,
   \quad
   \theta^T = (\, \theta^1, \theta^2 \,)\,,
   \quad
   \psi   = \begin{pmatrix} \psi^1 \\ \psi^2 \end{pmatrix}\,,
   \quad
   \psi^T = (\, \psi^1, \psi^2 \,)\,.
\end{equation}
where both $\theta^{\alpha}$ and $\psi^{\alpha}$ are \emph{real} Grassmann variables.
For Majorana spinors, the components $\theta^{\alpha}$ and $\psi^{\alpha}$ are labeled with \emph{upper} indices.  Different \emph{group} elements, $g_1,g_2,\dots$ are labeled with \emph{lower} indices.
From \eqref{e:Cdefinition}, charge conjugate spinors are given by:
\begin{equation}\label{e:C-GrassmannSpinorsDefs}
   C \theta = \begin{pmatrix} \theta^2 \\ -\theta^1 \end{pmatrix}\,,
   \quad
   \bar{\theta} \equiv \theta^T C = (\, -\theta^2,\theta^1 \,)\,,
   \quad
   C \psi = \begin{pmatrix} \psi^2 \\ -\psi^1 \end{pmatrix}\,,
   \quad
   \bar{\psi} \equiv \psi^T C = (\, -\psi^2,\psi^1 \,)\,.
\end{equation}
So for example,
\begin{equation}\label{e:GrassmannLengths}
   \bar{\theta} \theta = (\, -\theta^2,\theta^1 \,) \begin{pmatrix} \theta^1 \\ \theta^2 \end{pmatrix}
   =
   - \theta^2 \theta^1 + \theta^1 \theta^2
   =
   2 \theta^1 \theta^2 \>.
\end{equation}
Also,
\begin{equation}\label{e:barthet}
   \bar{\theta} \psi = (\, -\theta^2,\theta^1 \,) \begin{pmatrix} \psi^1 \\ \psi^2 \end{pmatrix}
   =
   -\theta^2 \psi^1 + \theta^1 \psi^2 \,,
   \qquad
   \bar{\psi} \theta = (\, -\psi^2,\psi^1 \,) \begin{pmatrix} \theta^1 \\ \theta^2 \end{pmatrix}
   =
   - \psi^2 \theta^1 + \psi^1 \theta^2
\end{equation}
Because $\theta^{\alpha} \psi^{\beta}$ anticommute, $\bar{\theta} \psi = \bar{\psi} \theta$.

For Majorana Grassmann integrals, we note that since
\begin{equation}\label{e:MajoranaDifferentials-deriv}
   \rd\theta = \begin{pmatrix} \rd\theta^1 \\ \rd\theta^2 \end{pmatrix}\,,
   \quad
   \rd\bar{\theta} = (\, \rd\theta^1, \rd\theta^2 \,) \, C = (\, -\rd\theta^2,\rd\theta^1 \,)\,,
\end{equation}
then
\begin{equation}\label{e:MajoranaDifferentials}
   \rd\bar\theta\,\rd\theta = 2 \, \rd\theta^{1} \rd\theta^2 \>.
\end{equation}
Then for example
\begin{equation}\label{e:berezintop}
   \int\rd^2\theta\,\theta^1\theta^2
   = \int\rd\theta^2\Bigl[\int\rd\theta^1\,\theta^1\theta^2\Bigr]
   = \int\rd\theta^2\,\theta^2 = 1\>.
\end{equation}

%
%
\section{Integration by parts in extended superspace}
\label{s:IBPextended}

In the $Z=0$ sector (\S\ref{ss:vacuum}), integrating $D_\alpha$ by parts inside $\int\rd^2x\,\rd^2\theta$ only ever produces an ordinary total $x$-derivative, since $D_\alpha=\partial_\alpha+i(C\gamma^\mu\theta)_\alpha\partial_\mu$ has just one ``spacetime-derivative-like'' piece. Once $\Psi$ is allowed to depend on $\xi$ (\S\ref{ss:Znotzero}), the extended $D_\alpha=\partial_\alpha+i(C\gamma^\mu\theta)_\alpha\partial_\mu-i(\gamma^1\theta)_\alpha\partial_\xi$ has a second such piece, and integrating by parts inside $\int\rd\ts=\int\rd x\,\rd^2\theta\,\rd\xi$ produces a genuinely new surface term in $\xi$ alongside the familiar one in $x$. This appendix proves the general result once, so it can be cited rather than re-derived each time it is needed.

\begin{theorem}[Integration by parts for $D_\alpha$ in extended superspace]
\label{thm:IBPextended}
For any superfield $X(x,\theta,\xi)=X_0+\theta^1X_1+\theta^2X_2+\theta^1\theta^2X_{12}$,
\begin{equation}\label{e:IBPextended-D1}
   \int\rd^2\theta\,D_1(X) = -i(\partial_0-\partial_1)\,X_2 \;+\; i\,\partial_\xi X_1\,,
\end{equation}
\begin{equation}\label{e:IBPextended-D2}
   \int\rd^2\theta\,D_2(X) = i(\partial_0+\partial_1)\,X_1 \;-\; i\,\partial_\xi X_2\,.
\end{equation}
Consequently,
\begin{equation}\label{e:IBPextended-full}
   \int\rd\ts\,D_1(X) = -i\!\int\!\rd\xi\,\Bigl[X_2\Bigr]_{x=-\infty}^{x=+\infty} + i\!\int\!\rd x\,\Bigl[X_1\Bigr]_{\xi\to-\infty}^{\xi\to+\infty}\,,
\end{equation}
\begin{equation}\label{e:IBPextended-full2}
   \int\rd\ts\,D_2(X) = i\!\int\!\rd\xi\,\Bigl[X_1\Bigr]_{x=-\infty}^{x=+\infty} - i\!\int\!\rd x\,\Bigl[X_2\Bigr]_{\xi\to-\infty}^{\xi\to+\infty}\,,
\end{equation}
each vanishing provided the bracketed component falls off at the corresponding boundary (or, for the $\xi$-boundary specifically, provided $\xi$ is instead taken compact and the component single-valued/periodic on it).
\end{theorem}

\begin{proof}
Write $D_\alpha=\partial_\alpha-i\theta^\alpha\bar\partial_\alpha-i\theta^{\bar\alpha}\partial_\xi$, where $\bar\partial_1\equiv\partial_0-\partial_1$, $\bar\partial_2\equiv\partial_0+\partial_1$, and $\bar\alpha$ denotes the index other than $\alpha$; this is \eqref{e:Ddef} rewritten using $(C\gamma^0\theta)_\alpha=-\theta^\alpha$, $(C\gamma^1\theta)_1=\theta^1$, $(C\gamma^1\theta)_2=-\theta^2$ (from \eqref{e:Cproperties}), and $(\gamma^1\theta)_\alpha=\theta^{\bar\alpha}$. Since $\int\rd^2\theta\,[\cdots]$ extracts the $\theta^1\theta^2$ coefficient (\S\ref{s:grassmann}), each of the three pieces of $D_\alpha(X)$ is handled separately.

\emph{The $\partial_\alpha$ piece.} $\int\rd^2\theta\,\partial_\alpha X=0$ identically: differentiating with respect to a Grassmann variable always lowers the $\theta$-degree by one, so $\partial_\alpha X$ has no $\theta^1\theta^2$ term for it to contribute.

\emph{The $-i\theta^\alpha\bar\partial_\alpha$ piece.} For $\alpha=1$: $\theta^1X=\theta^1X_0+\theta^1\theta^2X_2$ (the $\theta^1X_1$ and $\theta^1\theta^2X_{12}$ terms die against the repeated $\theta^1$), so its $\theta^1\theta^2$ coefficient is $X_2$, giving $-i\bar\partial_1X_2$. For $\alpha=2$: $\theta^2X=\theta^2X_0-\theta^1\theta^2X_1$ (using $\theta^2\theta^1=-\theta^1\theta^2$), so its $\theta^1\theta^2$ coefficient is $-X_1$, giving $-i\bar\partial_2(-X_1)=i\bar\partial_2X_1$.

\emph{The $-i\theta^{\bar\alpha}\partial_\xi$ piece.} Identical to the previous piece with $\bar\partial_\alpha\to\partial_\xi$ and $\alpha\to\bar\alpha$: for $\alpha=1$ ($\bar\alpha=2$) this is $-i\partial_\xi\times(-X_1)=i\partial_\xi X_1$; for $\alpha=2$ ($\bar\alpha=1$) this is $-i\partial_\xi X_2$.

Summing the surviving pieces for each $\alpha$ gives \eqref{e:IBPextended-D1}--\eqref{e:IBPextended-D2}. Integrating these over $x$ and $\xi$ turns each total derivative into the boundary term evaluated at the corresponding limits, giving \eqref{e:IBPextended-full}--\eqref{e:IBPextended-full2}; a boundary term vanishes if the bracketed quantity falls off at that boundary, or automatically if the corresponding coordinate is instead compact and the quantity is single-valued on it (the total derivative then integrates to zero over one period).
\end{proof}

\begin{corollary}
\label{cor:IBPextended}
For $X=(D_1\Psi)(\delta\Psi)$, a direct computation from \eqref{e:Ddef} and \eqref{e:Phi-def} (the same method as \eqref{e:D1D2Psi}, now keeping $\partial_\xi\phi,\partial_\xi\psi\neq0$) gives $D_1\Psi=\psi^2-i\theta^1(\partial_0-\partial_1)\phi+i\theta^2(F-\partial_\xi\phi)+i\theta^1\theta^2[\cdots]$, so with $\delta\Psi=\delta\phi-\theta^2\delta\psi^1+\theta^1\delta\psi^2+i\theta^1\theta^2\delta F$,
\begin{equation}\label{e:X2-corollary}
   X\big|_{\theta^2} = \psi^2\,\delta\psi^1 + i\bigl(F-\partial_\xi\phi\bigr)\delta\phi\,.
\end{equation}
By \eqref{e:IBPextended-full2}, the $\xi$-boundary piece of $\int\rd\ts\,D_2\bigl[(D_1\Psi)(\delta\Psi)\bigr]$ is therefore
\begin{equation}\label{e:z1.surfaceterm}
   -i\int\rd x\,\Bigl[X\big|_{\theta^2}\Bigr]_{\xi\to-\infty}^{\xi\to+\infty}
   =
   \int\rd x\,\Bigl[\bigl(F-\partial_\xi\phi\bigr)\delta\phi - i\,\psi^2\,\delta\psi^1\Bigr]_{\xi\to-\infty}^{\xi\to+\infty}\,,
\end{equation}
while the $x$-boundary piece is identical in form to the $Z=0$ case and vanishes under the same assumption, $\delta\Psi\to0$ as $x\to\pm\infty$.
\end{corollary}

\bibliography{WZmodelLecture}

\begin{thebibliography}{10}%
\makeatletter
\providecommand \@ifxundefined [1]{%
 \@ifx{#1\undefined}
}%
\providecommand \@ifnum [1]{%
 \ifnum #1\expandafter \@firstoftwo
 \else \expandafter \@secondoftwo
 \fi
}%
\providecommand \@ifx [1]{%
 \ifx #1\expandafter \@firstoftwo
 \else \expandafter \@secondoftwo
 \fi
}%
\providecommand \natexlab [1]{#1}%
\providecommand \enquote  [1]{``#1''}%
\providecommand \bibnamefont  [1]{#1}%
\providecommand \bibfnamefont [1]{#1}%
\providecommand \citenamefont [1]{#1}%
\providecommand \href@noop [0]{\@secondoftwo}%
\providecommand \href [0]{\begingroup \@sanitize@url \@href}%
\providecommand \@href[1]{\@@startlink{#1}\@@href}%
\providecommand \@@href[1]{\endgroup#1\@@endlink}%
\providecommand \@sanitize@url [0]{\catcode `\\12\catcode `\$12\catcode
  `\&12\catcode `\#12\catcode `\^12\catcode `\_12\catcode `\%12\relax}%
\providecommand \@@startlink[1]{}%
\providecommand \@@endlink[0]{}%
\providecommand \url  [0]{\begingroup\@sanitize@url \@url }%
\providecommand \@url [1]{\endgroup\@href {#1}{\urlprefix }}%
\providecommand \urlprefix  [0]{URL }%
\providecommand \Eprint [0]{\href }%
\providecommand \doibase [0]{https://doi.org/}%
\providecommand \selectlanguage [0]{\@gobble}%
\providecommand \bibinfo  [0]{\@secondoftwo}%
\providecommand \bibfield  [0]{\@secondoftwo}%
\providecommand \translation [1]{[#1]}%
\providecommand \BibitemOpen [0]{}%
\providecommand \bibitemStop [0]{}%
\providecommand \bibitemNoStop [0]{.\EOS\space}%
\providecommand \EOS [0]{\spacefactor3000\relax}%
\providecommand \BibitemShut  [1]{\csname bibitem#1\endcsname}%
\let\auto@bib@innerbib\@empty
\bibitem [{\citenamefont {Wess}\ and\ \citenamefont
  {Zumino}(1974)}]{WessZumino1974}%
  \BibitemOpen
  \bibfield  {author} {\bibinfo {author} {\bibfnamefont {J.}~\bibnamefont
  {Wess}}\ and\ \bibinfo {author} {\bibfnamefont {B.}~\bibnamefont {Zumino}},\
  }\bibfield  {title} {\bibinfo {title} {Supergauge transformations in four
  dimensions},\ }\href {https://doi.org/10.1016/0550-3213(74)90355-1}
  {\bibfield  {journal} {\bibinfo  {journal} {Nucl. Phys. B}\ }\textbf
  {\bibinfo {volume} {70}},\ \bibinfo {pages} {39} (\bibinfo {year}
  {1974})}\BibitemShut {NoStop}%
\bibitem [{\citenamefont {Witten}\ and\ \citenamefont
  {Olive}(1978)}]{WittenOlive1978}%
  \BibitemOpen
  \bibfield  {author} {\bibinfo {author} {\bibfnamefont {E.}~\bibnamefont
  {Witten}}\ and\ \bibinfo {author} {\bibfnamefont {D.}~\bibnamefont {Olive}},\
  }\bibfield  {title} {\bibinfo {title} {Supersymmetry algebras that include
  topological charges},\ }\href {https://doi.org/10.1016/0370-2693(78)90357-X}
  {\bibfield  {journal} {\bibinfo  {journal} {Phys. Lett. B}\ }\textbf
  {\bibinfo {volume} {78}},\ \bibinfo {pages} {97} (\bibinfo {year}
  {1978})}\BibitemShut {NoStop}%
\bibitem [{\citenamefont {Gates}\ \emph {et~al.}(1983)\citenamefont {Gates},
  \citenamefont {Grisaru}, \citenamefont {Ro{\v{c}}ek},\ and\ \citenamefont
  {Siegel}}]{GatesGrisaruRocekSiegel1983}%
  \BibitemOpen
  \bibfield  {author} {\bibinfo {author} {\bibfnamefont {J.}~\bibnamefont
  {Gates}, \bibfnamefont {S.~James}}, \bibinfo {author} {\bibfnamefont {M.~T.}\
  \bibnamefont {Grisaru}}, \bibinfo {author} {\bibfnamefont {M.}~\bibnamefont
  {Ro{\v{c}}ek}},\ and\ \bibinfo {author} {\bibfnamefont {W.}~\bibnamefont
  {Siegel}},\ }\href@noop {} {\emph {\bibinfo {title} {Superspace, or One
  Thousand and One Lessons in Supersymmetry}}},\ \bibinfo {series} {Frontiers
  in Physics}, Vol.~\bibinfo {volume} {58}\ (\bibinfo  {publisher}
  {Benjamin/Cummings},\ \bibinfo {address} {Reading, MA},\ \bibinfo {year}
  {1983})\ \bibinfo {note} {free online edition with corrections and bookmarks
  at \url{https://arxiv.org/abs/hep-th/0108200}},\ \Eprint
  {https://arxiv.org/abs/hep-th/0108200} {arXiv:hep-th/0108200} \BibitemShut
  {NoStop}%
\bibitem [{\citenamefont {Wess}\ and\ \citenamefont
  {Bagger}(1992)}]{WessBagger1992}%
  \BibitemOpen
  \bibfield  {author} {\bibinfo {author} {\bibfnamefont {J.}~\bibnamefont
  {Wess}}\ and\ \bibinfo {author} {\bibfnamefont {J.}~\bibnamefont {Bagger}},\
  }\href@noop {} {\emph {\bibinfo {title} {Supersymmetry and Supergravity}}},\
  \bibinfo {edition} {2nd}\ ed.\ (\bibinfo  {publisher} {Princeton University
  Press},\ \bibinfo {address} {Princeton, NJ},\ \bibinfo {year}
  {1992})\BibitemShut {NoStop}%
\bibitem [{\citenamefont {Bogomolny}(1976)}]{Bogomolny1976}%
  \BibitemOpen
  \bibfield  {author} {\bibinfo {author} {\bibfnamefont {E.~B.}\ \bibnamefont
  {Bogomolny}},\ }\bibfield  {title} {\bibinfo {title} {Stability of classical
  solutions},\ }\href@noop {} {\bibfield  {journal} {\bibinfo  {journal} {Sov.
  J. Nucl. Phys.}\ }\textbf {\bibinfo {volume} {24}},\ \bibinfo {pages} {449}
  (\bibinfo {year} {1976})},\ \bibinfo {note} {[Yad. Fiz. \textbf{24}, 861
  (1976)]}\BibitemShut {NoStop}%
\bibitem [{\citenamefont {Prasad}\ and\ \citenamefont
  {Sommerfield}(1975)}]{PrasadSommerfield1975}%
  \BibitemOpen
  \bibfield  {author} {\bibinfo {author} {\bibfnamefont {M.~K.}\ \bibnamefont
  {Prasad}}\ and\ \bibinfo {author} {\bibfnamefont {C.~M.}\ \bibnamefont
  {Sommerfield}},\ }\bibfield  {title} {\bibinfo {title} {Exact classical
  solution for the 't hooft monopole and the julia-zee dyon},\ }\href
  {https://doi.org/10.1103/PhysRevLett.35.760} {\bibfield  {journal} {\bibinfo
  {journal} {Phys. Rev. Lett.}\ }\textbf {\bibinfo {volume} {35}},\ \bibinfo
  {pages} {760} (\bibinfo {year} {1975})}\BibitemShut {NoStop}%
\bibitem [{\citenamefont {Jackiw}\ and\ \citenamefont
  {Rebbi}(1976)}]{JackiwRebbi1976}%
  \BibitemOpen
  \bibfield  {author} {\bibinfo {author} {\bibfnamefont {R.}~\bibnamefont
  {Jackiw}}\ and\ \bibinfo {author} {\bibfnamefont {C.}~\bibnamefont {Rebbi}},\
  }\bibfield  {title} {\bibinfo {title} {Solitons with fermion number 1/2},\
  }\href {https://doi.org/10.1103/PhysRevD.13.3398} {\bibfield  {journal}
  {\bibinfo  {journal} {Phys. Rev. D}\ }\textbf {\bibinfo {volume} {13}},\
  \bibinfo {pages} {3398} (\bibinfo {year} {1976})}\BibitemShut {NoStop}%
\bibitem [{\citenamefont {Appelquist}\ \emph {et~al.}(2001)\citenamefont
  {Appelquist}, \citenamefont {Cheng},\ and\ \citenamefont
  {Dobrescu}}]{AppelquistChengDobrescu2001}%
  \BibitemOpen
  \bibfield  {author} {\bibinfo {author} {\bibfnamefont {T.}~\bibnamefont
  {Appelquist}}, \bibinfo {author} {\bibfnamefont {H.-C.}\ \bibnamefont
  {Cheng}},\ and\ \bibinfo {author} {\bibfnamefont {B.~A.}\ \bibnamefont
  {Dobrescu}},\ }\bibfield  {title} {\bibinfo {title} {Bounds on universal
  extra dimensions},\ }\href {https://doi.org/10.1103/PhysRevD.64.035002}
  {\bibfield  {journal} {\bibinfo  {journal} {Phys. Rev. D}\ }\textbf {\bibinfo
  {volume} {64}},\ \bibinfo {pages} {035002} (\bibinfo {year} {2001})},\
  \Eprint {https://arxiv.org/abs/hep-ph/0012100} {arXiv:hep-ph/0012100}
  \BibitemShut {NoStop}%
\bibitem [{\citenamefont {Servant}\ and\ \citenamefont
  {Tait}(2003)}]{ServantTait2003}%
  \BibitemOpen
  \bibfield  {author} {\bibinfo {author} {\bibfnamefont {G.}~\bibnamefont
  {Servant}}\ and\ \bibinfo {author} {\bibfnamefont {T.~M.~P.}\ \bibnamefont
  {Tait}},\ }\bibfield  {title} {\bibinfo {title} {Is the lightest
  kaluza--klein particle a viable dark matter candidate?},\ }\href
  {https://doi.org/10.1016/S0550-3213(02)01012-X} {\bibfield  {journal}
  {\bibinfo  {journal} {Nucl. Phys. B}\ }\textbf {\bibinfo {volume} {650}},\
  \bibinfo {pages} {391} (\bibinfo {year} {2003})},\ \Eprint
  {https://arxiv.org/abs/hep-ph/0206071} {arXiv:hep-ph/0206071} \BibitemShut
  {NoStop}%
\bibitem [{\citenamefont {{Anthropic}}(2026)}]{ClaudeSonnet5}%
  \BibitemOpen
  \bibfield  {author} {\bibinfo {author} {\bibnamefont {{Anthropic}}},\
  }\href@noop {} {\bibinfo {title} {{C}laude {S}onnet 5}},\ \bibinfo
  {howpublished} {Large language model} (\bibinfo {year} {2026}),\ \bibinfo
  {note} {model id \texttt{claude-sonnet-5}; used October 2026.
  \url{https://www.anthropic.com}}\BibitemShut {NoStop}%
\end{thebibliography}%
\end{document}